\documentclass[runningheads]{llncs}
\usepackage{booktabs} % For formal tables
\usepackage[ruled]{algorithm2e} % For algorithms

\usepackage{mathtools}
\usepackage{amsmath}
\usepackage{amsfonts}
\usepackage{bbm}
\usepackage{xspace}
\usepackage{graphicx}
\usepackage{appendix}
\usepackage{listings}
\usepackage[normalem]{ulem} % for strikethrough
\usepackage{xcolor}

\usepackage{amsthm}
\usepackage{enumitem}
\usepackage{doi}

\newtheorem{observation}{Observation}

\renewcommand{\cal}[1]{\ensuremath{\mathcal{#1}}}
\renewcommand{\sf}[1]{\ensuremath{\mathsf{#1}}}
\newcommand{\abs}[1]{\left| #1 \right| }
\newcommand{\norm}[2]{\left\Vert #2 \right\Vert_{#1} }
\newcommand{\inp}[1]{\left( #1 \right) }
\newcommand{\inb}[1]{\left[ #1 \right] }

\newcommand{\inset}[1]{\left\{ #1 \right\} }
\newcommand{\inbr}[1]{\left\{ #1 \right\}}

\newcommand{\vv}[1]{\ensuremath{\boldsymbol{#1}}}
\newcommand{\E}{\ensuremath{\mathop{\mathbb{E}}}}
\newcommand{\rohit}[1]{{\color{blue} #1}}

\newcommand{\vrep}{\ensuremath{\sf{R}\xspace}}
\newcommand{\vvrep}{\ensuremath{\vv{\sf{R}}\xspace}}
\newcommand{\Z}{\ensuremath{\mathbb{Z}}}

\usepackage{hyperref}

\usepackage{cleveref}

\newcommand{\cancel}[1]{\ignorespaces}
\newcommand{\vns}{\ensuremath{m}\xspace}
\newcommand{\vnc}{\ensuremath{n}\xspace}
\newcommand{\vrsc}{\ensuremath{\rho}\xspace}

\newcommand{\vgame}{\ensuremath{\mathcal{G}}\xspace}
\newcommand{\vtype}{\ensuremath{r}\xspace}

\title{Game of Trust: How Trustworthy Does Your\\ Blockchain Think You Are?} 
 \author{Petros Drineas\inst{1} \and Rohit Nema\inst{2} \and Rafail Ostrovsky\inst{3} \and Vassilis Zikas\inst{4}}
 \institute{Purdue University \email{pdrineas@purdue.edu}
 \and
 Stanford University \email{rnema@cs.stanford.edu}
 \and
 UCLA \email{rafail@cs.ucla.edu}
 \and
 Georgia Tech \email{vzikas@gatech.edu}}

\begin{document}

\maketitle

\begin{abstract}

%\rohit{UPDATE ABSTRACT.}
We investigate how a blockchain can distill the collective belief of its nodes regarding the trustworthiness of a (sub)set of nodes into a {\em reputation system} that reflects the probability of correctly performing a task. To address this question, we introduce a framework that breaks it down into two sub-problems:

{\em 1. (Information Extraction):} How can the system distill trust information from a function of the nodes' true beliefs?
 
 {\em 2. (Incentive Design):} How can we incentivize nodes to truthfully report such information?

To tackle the first sub-problem, we adapt, in a non-trivial manner, the well-known PageRank algorithm to our problem. For the second, we define a new class of games, called Trustworthy Reputation games (TRep games), which aim to extract the collective beliefs on trust from the actions of rational participants.
We then propose a concrete TRep game whose utility function leverages Personalized PageRank and can be instantiated through a straightforward blockchain rewards mechanism. Building on this, we show how the TRep game enables the design of a reputation system. Such systems can  enhance the robustness, scalability, and efficiency of blockchain and DeFi solutions. For instance, we demonstrate how such a system can be used within a Proof-of-Reputation blockchain.
\end{abstract}

\section{Introduction}
The concept of trust is front and center in any decentralized system, such as blockchain ledgers and cryptocurrency. Being able to take an informed decision on which nodes to trust to perform a certain task can vastly improve the quality of a solution and eliminate bottlenecks.  

Consider, for example, %To demonstrate this we use are a running example 
the 
{\em optimistic} (protocols) paradigm: % that has dominated the recent blockchain, cryptocurrency, and DeFi ecosystem: 
At a high level, such protocols rely on assumptions about the trustworthiness of specific entities within the blockchain ecosystem in order to circumvent inefficiencies that arise from treating all participants as potentially adversarial. This %optimistic
paradigm has led to some of the most prominent Layer-2 scaling solutions, such as optimistic rollups (e.g. \cite{buterin2019optimistic,landis2023incentive,moosavi2023fast,capretto2024fast,ye2022specular,lee2025looking}). In the extreme, such protocols may require unconditional trust in designated entities. However, in systems that secure assets worth billions of dollars, absolute trust is problematic and poses systemic risks.

To mitigate this, one %optimistic designs are often strengthened by 
often replaces the assumptions of absolute trust with more plausible assumptions about incentives that guide participants behavior. Examples of this approach include payment channels (e.g. ~\cite{kolachala2024sok,avarikioti2019ride,avarikioti2020cerberus,avarikioti2021brick,ersoy2020profit}), payment networks (e.g. \cite{poon2016lightning,sprites2017miller}), atomic swaps (e.g. \cite{herlihy2018atomic}),  automated market makers (e.g. \cite{angeris2020improved,angeris2022foundations}), stake delegation (e.g. \cite{kiayias2017ouroboros,boneh2020proof}), and many others. Broadly speaking, these mechanisms share a common theme: the behavior of (presumably) trusted parties is collateralized, such that any deviation from honest execution risks a significant financial penalty. Under the assumption that rational actors seek to preserve their collateral, this incentive structure provides a practical safeguard for protocol correctness and reliability.

\cancel{ 
%let's uncancel this for eprint
Collateralization has emerged as a powerful mechanism in modern blockchain, cryptocurrency, and decentralized finance (DeFi) applications. Nonetheless, it imposes significant costs. Requiring participants to lock substantial amounts of assets results in inefficient capital utilization, diminishes the liquidity of the underlying tokens, raises barriers to entry, and makes the rotation of trusted parties inherently slow. Consequently, such systems tend to favor well-capitalized actors, thereby increasing the risk of (re-)centralization. Stake delegation provides a concrete example of this dynamic: typically, only large and well-capitalized stake pools operate profitably; delegation often takes effect only after a delay once it is recorded on-chain; and participants are generally required to commit their stake for an extended period, further constraining flexibility.

Moreover, collateral-based incentive mechanisms couple honest behavior to the preservation of a static asset---the locked collateral. Misbehavior threatens this fixed stake, which otherwise remains idle, serving only as a guarantee of correct behavior. By contrast, reward/revenue-based  mechanisms (also referred to as reward-sharing) align honesty with the prospect of ongoing gains. In such systems, actors are motivated not only by the avoidance of loss but also by the preservation of future income streams, which accrue continuously through sustained honest participation.
}

An alternative approach to determining whom to trust---without relying solely on collateralization---is to base trust decisions on a reputation system (also commonly referred to as a recommender system). {\em Reputation systems} have been extensively studied and applied across computer science, social science, and engineering. In their abstract form, such systems assign each party a {\em trustworthiness score} that reflects the extent to which the party can be trusted in a given context. For example, in a blockchain setting, the score may capture the likelihood that a node will follow the prescribed protocol.
Reputation can significantly improve decision-making about whom to trust, and it is widely used on the internet and in social networks. Prominent examples include Yelp, TripAdvisor, Google, and Amazon, all of which provide reputation mechanisms. % as core functionality. 

Several existing blockchains and cryptocurrencies implicitly---and often informally ---leverage the perceived trustworthiness of nodes to enhance security, scalability, and decentralization. For example, Ripple XRP’s consensus mechanism \cite{schwartz2014ripple} assumes that nodes only accept messages from a predefined set of trusted servers; Ethereum 2.0 \cite{ethereum2.0} incentivizes honest behavior and penalizes dishonesty to secure the network and accelerate block production; and Algorand \cite{algorand} rewards participation to increase the likelihood that selected committees contain a sufficient number of members.
Beyond these implicit mechanisms, a number of proposals explicitly demonstrate that access to trust information---often formalized as a reputation vector (see below)---can improve fundamental blockchain properties such as liveness, safety, and scalability, via the paradigm commonly referred to as Proof of Reputation (PoR)~\cite{Gai2018ProofOR,CANS:Chow07,8645706,EPRINT:BirFehKho17,KleOstZik20}. Similar reasoning applies to decentralized finance (DeFi), where the collective belief of the system in a node’s ability to repay blockchain loans can be interpreted as the system’s assessment of that node’s credit-worthiness.
However, to the best of our knowledge, reputation systems in these contexts are typically assumed as setup assumptions, or else are ad hoc and implicit, and they lack rigorous formal analysis.
The central question motivating our work is therefore:
\begin{quote}
How can a blockchain derive a reputation system for (a subset of) its nodes by distilling the collective belief of the system regarding the trustworthiness of these nodes?
\end{quote}
Prior work has considered related questions, such as computing a reputation vector for blockchain nodes~\cite{DBLP:conf/icisc/Larangeira22,8645706}. However, these approaches equate reputation with system resources---such as stake or hashing power---and thus do not directly address the above question.

\subsection{Our Contributions}\label{sec:contributions}
Our work makes conceptual, definitional, and technical contributions, towards  answering the above question. 

On a conceptual level, we dissect the above question into two sub-problems:  
\begin{itemize}
\item {\bf Problem 1.} Given access to (a function of) each node's belief about the trustworthiness of other nodes $s_1, \ldots, s_\vns$---i.e., probability of each honestly executing a given task---how can we distill a ``meaningful'' system-wide reputation (score) $\vrsc_{i}$, that reflects the system's belief about $s_i$'s trustworthiness?
\item {\bf Problem 2.} How can we incentivize the nodes to truthfully report the above desired function of their belief?  
\end{itemize}

Naturally, an answer to the above two problems provides an answer to our original motivating question: under the assumption of rational nodes, the blockchain can learn the desired function of their (true) beliefs and use it to distill a meaningful system-wide reputation. As such, our result can be view as a systematic framework that can be used to address our   above overarching goal. % that motivates our work. 
% Here is how we attack each of the above subproblems. \rohit{maybe remove this line?} 

For the first problem we observe that it is in some sense isomorphic to (the weighted version of) a well-studied problem in computer science: Given a graph that encodes the Web (nodes are pages and directed edges correspond to links) how can we distill the importance of each node/page in this graph in a single number? The celebrated result that gave a solution to the above was the well known PageRank algorithm~\cite{pagerank}. Originally proposed by S.~Brin and L.~Page and credited for the indisputable success of Google's search engine, PageRank allows to use the Web-graph---where the pages are nodes and links are directed edges---to extract a normalized score for each web-page---i.e., so that all scores sum up to $1$. Intuitively, the score (aka Rank or PageRank) of each page corresponds to its importance in the Web. The weighted version of PageRank~\cite{1344743} allows pages/nodes to annotate their links according to (their belief of) the relative importance of each page they link to. As such, (weighted) PageRank has been  proposed in the literature to derive reputation systems, e.g.,~\cite{Zhang28052017,articlePRR}, recommender systems, e.g.,~\cite{inbook,10.1145/1390334.1390465}, and to estimate how much trust to allocate on services and nodes in different networked scenarios ({\em cf}. ~\cite{10.5555/1316689.1316740} and references therein). A straw-man approach here is to apply it also to solve Problem~1.   

To anyone familiar with the use of PageRank and the problems it solves the above idea appears extremely natural. And the ``meaningfulness'' of its output is attested in practice by the vast success of its numerous applications. But, is it possible to offer provable evidence of the above? 

In its most generality, the answer is, no---What does it mean that a number meaningfully reflects the system's belief? Notwithstanding, due to the nature of our problem, we are able to put forth a test of meaningfulness: Assuming that the trustworthiness score of each ranked node $s_i$ is $\vrep_{j}$ can we prove that PageRank outputs it? 

Even the question above is subtler than it first appears. When the nodes’ beliefs deviate significantly from the ground truth (i.e., the true trustworthiness scores of the ranked nodes), it becomes information-theoretically impossible to recover even simple correlations. In such cases, PageRank must be regarded as a heuristic---albeit one strongly justified by overwhelming empirical evidence of its effectiveness in practice.

This observation, however, leads naturally to a meaningful benchmark for evaluating the suitability of PageRank in our setting: when the nodes’ beliefs are accurate, the output of PageRank yields a faithful encoding of the servers’ trustworthiness scores. We view such benchmarks as an interesting contribution, as we are not aware of any other such benchmarks for problems PageRank is applied to. Moreover, in the course of proving this statement, we uncover several intriguing tuning and interpretations of the PageRank algorithm. As such,  we believe that our re-purposing and treatment of PageRank for our problem can be of  interest even beyond Problem 1. We provide more details in the technical overview (Section~\ref{sec:tech_overview}).

Armed with a solution to Problem 1, we turn to Problem 2. 
To address this challenge, we introduce and formally define a new class of games, which we call {\em Trustworthy Reputation (TRep) games.} Intuitively, these games are constructed so that it is a Nash equilibrium for players to play/report a mixed strategy that encodes precisely the same information extracted by our fine-tuned PageRank algorithm. Consequently, when players report their actions according to their Nash strategy  (which we show is even unique in certain cases), the resulting profile can be used to recover the same meaningful reputation values as the solution to Problem 1 (which, recall, assumed truthful reporting of beliefs).

We instantiate our definition with a concrete utility function, and prove that our instantiation achieves the specified properties. 
A notable feature of the above utility function,  in addition to inducing the desired Nash equilibrium, is that it yields an equitable allocation of blockchain rewards to participating nodes.
In the technical portion of the paper we establish these results under different assumptions about the accuracy of nodes’ beliefs regarding the trustworthiness of ranked servers. Moreover, we demonstrate that embedding such a game into a proof-of-reputation blockchain (e.g., \cite{KleOstZik20}), and calibrating its rewards according to our utility function, provides a solution to the second subproblem: namely, it becomes a Nash equilibrium for users to truthfully record their encoded assessments of nodes’ trustworthiness on the blockchain. %\footnote{We prove this to be a Nash equilibrium when users are perfectly informed, and an $\epsilon$-Nash equilibrium when users hold noisy but consistent beliefs (see Section \ref{section:contributions_noisy_information}).}

To our knowledge, our work is the first to integrate formal  game-theoretic reasoning with reputation-extraction mechanisms in order to provide a formal, theory-grounded answer to the question of whether a blockchain can reliably distill a reputation system for (a subset of) its nodes. Given the current state of the field, we view even the precise formulation of this mathematical and computer-science problem as a substantial step forward. Finally, we emphasize that our formulation of trustworthy reputation games is generic: it is not restricted to PageRank or even to blockchain settings. We conjecture that it can also be instantiated with alternative methods, thereby extending its applicability well beyond the blockchain ecosystem.

\subsection{Related Work}\label{section:related}

A number of works have studied reputation as an important concept of behavioral analysis in game theory, (see, e.g.,~\cite{repinseqgames,DBLP:phd/ndltd/Mui02,sung} and references therein.) One of the most important application domains is in repeated and more, generally, sequential games, where reputation typically models  information  about a  player's state that can be extracted from the player's past actions and used to improve the future response to this player~\cite{mailath2006repeated,Abreu2007,BadRep,Fudenberg1992,RepWithConflict}. A related line of work investigates learning nature in repeated games~\cite{RENAULT2004124,Leoni2014,Sugaya2020,Cripps2008,Fudenberg2011}. Our work can be seen as combining both threads in a novel way: our goal is to learn the state of nature, which corresponds to an external reputation (as opposed to learning the reputation of the players themselves) by using the strategies of the player (as opposed to the move of nature). Moreover, we aim to learn this state in a one-shot game  which is motivated by a blockchain application. 

In the blockchain context, a number of proposals have explored leveraging existing reputation systems to enhance key properties of consensus and blockchain protocols, such as liveness, finality, and scalability~\cite{Gai2018ProofOR,CANS:Chow07,8645706,EPRINT:BirFehKho17,KleOstZik20}. These approaches typically focus on incorporating reputation but often do not fully address how such a system can be derived in a way that reflects the actual trustworthiness---i.e., the probability of honesty---of participating nodes. 

In many cases, reputation is operationalized through token-based mechanisms, such as awarding reputation tokens for hashing or staking~\cite{DBLP:conf/icisc/Larangeira22,8645706}. (See~\cite{blogpost} for a broader discussion of the economic aspects of tokenizing reputation.) For instance,~\cite{DBLP:conf/icisc/Larangeira22} describes a model where nodes may delegate their stake, and a node’s reputation is defined as the fraction of total delegated stake it receives.

Our work takes a different perspective. We focus on how a blockchain might derive information about node trustworthiness directly from the beliefs of its participants, rather than indirectly through measures such as delegated stake or computational power.

This distinction is important: while reputation may at times correlate with factors like delegated stake or hashing power, these metrics do not necessarily capture the broader and more nuanced notion of trustworthiness in performing specific tasks. In fact, for certain tasks the connection may be minimal. As a result, approaches that equate reputation solely with stake delegation~\cite{DBLP:conf/icisc/Larangeira22} or computational work~\cite{8645706} provide only a partial picture of reputation and leave open the question we address here.

Similarly, mechanisms such as slashing and stake 
delegation in delegated Proof-of-Stake protocols~\cite{DBLP:conf/crypto/KiayiasRDO17,DBLP:conf/eurosp/BrunjesKKS20,algorand,ethereum2.0} can be viewed as associating trustworthiness with the amount of stake owned, collateralized, or delegated---rather than explicitly capturing the system’s collective beliefs about node behavior.

%To our knowledge, none of the model proposed in previous works accurate captures the goals of our application. % for our setting. %, which is also motivate by a blockchain scenario. 
Several works have applied game-theoretic reasoning to the PageRank algorithm~\cite{NetworkReputation,7525508,foulley2018}; however, the goal of these works is rooted in the original use of the algorithm to discover important pages on the Internet and/or important nodes in a communication network. In particular, the goal is for the players (page creators) to maximize some revenue (or a network related quantity like latency of a transmitted message or accuracy of anomaly detection) by strategically pointing to appropriate pages/network-nodes. 
% \rohit{TODO: rework to talk about tiered graph}
While we also examine the effect of strategic play on PageRank, we do so in a specially designed game where (Personalized) PageRank is also used in the utilities. Our goal is to instead show that PageRank captures reputation in an extremely natural way and playing strategically indeed, recovers a meaningful quantity (the relative trustworthiness scores). Discovering other classes of graphs, which enable such reasoning is in our opinion a very interesting direction.

It is also useful to distinguish between our question of  how a blockchain can distill the beliefs of its nodes, and the related question of  how nodes form such local or personal beliefs in the first place. These questions are orthogonal, and each is interesting in its own right.

The latter---how nodes establish their own beliefs---is closely connected to several well-studied areas in machine learning, artificial intelligence, and game theory/economics. These fields provide frameworks for predicting the future behavior of others based on local information. Examples include multi-agent reinforcement learning (see~\cite{10.1007/s10462-021-09996-w} for a comprehensive survey), behavioral prediction (see~\cite{10.24963/ijcai.2023/746}), and the use of reputation in repeated games (see~\cite{mailath2006repeated,repinseqgames,DBLP:phd/ndltd/Mui02}). This extensive literature offers a wealth of methods---often complementary, sometimes competing---that blockchain nodes could use to form local beliefs about the trustworthiness of others.

However, these works do not resolve our central question: how can the blockchain itself use such local beliefs to address protocol-level challenges---for example, rewarding honesty and availability, or selecting honest-majority committees in Proof-of-Reputation (PoR) based iterated Byzantine fault-tolerant (iBFT) constructions like~\cite{KleOstZik20}. 
This is the concrete problem our work tackles.

Importantly, answering this question is non-trivial---even assuming solutions to belief-formation exist---because tasks like the above require the blockchain protocol to converge to a system-wide estimate of trustworthiness, distilled from individual nodes’ beliefs and agreed upon by all participants.

Several works have adapted PageRank to meet the needs of specific applications. Topic-Sensitive PageRank (or Topic PageRank)~\cite{10.1145/511446.511513} is one such refinement, designed to incorporate context through topics. The idea is simple yet powerful: rather than assuming that the random surfer teleports uniformly to any page in the graph (as in standard PageRank), Topic PageRank biases this teleportation toward a set of nodes relevant to a chosen topic.

Each topic is defined by a personalization vector---a probability distribution over nodes---which determines the likelihood of teleportation. When the random surfer teleports, they are therefore more likely to land on topic-relevant nodes. The resulting PageRank values capture importance with respect to that topic, rather than globally.

Although not directly aligned with our setting---since our graph already corresponds to the specific topic of trustworthiness of ranked nodes---our approach of restricting teleportation to the set of potential endorsees can be viewed as a related adaptation. More broadly, it illustrates how carefully tuning the teleportation mechanism and personalization parameters enables PageRank to extend well beyond its original scope.

%To our knowledge, however, the above works do not investigate the question of extracting trustworthiness from a moves of in a one shot game

%To our knowledge, however, these works do not deal with the problem of extracting re

Finally, the structure and goals of trustworthy reputation games bare resemblance to Bayesian optimal design (BOD)~\cite{10.5555/1296179}. However there are key differences: in contrast to BOD, (1) we (the designer) do not know the distribution of the agent's valuation (types), and (2) our objective is to learn the ground truth (nature's state), rather than these types or their distribution. Similarly, to our knowledge, techniques from  {\em Prior-independent mechanisms} (PIMs)~\cite{tim-optimal-mech,pimp-mech,AZAR2019511}, which  are designed for the incomplete information setting, do not apply here as the objectives are different and the game is a one-shot game.  

\section{Preliminaries}
\label{sec:prelim}
We first introduce the basic notation used throughout our technical sections. Then, we provide the relevant background and notation on the PageRank and Personalized PageRank algorithm.
%, and (reputation-based) blockchains. 

%\subsection{Notation}
Let $\Z$ denote the set of integers, and $\Z_{\geq k}$, the set of integers greater than or equal to $k$. For $n \in \Z_{\geq 1}$, let $[n]$ equal the set, $\inset{1, 2, \ldots, n}$.

For any probabilistic event $E$, let the indicator of $E$, $\mathbf{1}\inb{E}$ denote the binary random variable that outputs 1 when $E$ {\em occurs} and 0 otherwise.

We write vectors in boldface, as in $\vv{v}, \vv{\pi}, \ldots$; and matrices in capital letters. Denote the $n\times m$ matrix of all ones as, $\vv{1}_{n\times m}$. Then, $\vv{1}_{n\times 1}$ is the row-vector of all one's of length $n$.
Often when referring to the coordinates of some vector, $\vv{v}$, we drop the boldface. Therefore, the $i$th coordinate of $\vv{v}$ is $v_i$. We index matrices using $M[i,j]$ to denote the value in $i$th row and $j$th column.
Let $\vv{e_i}$ denote the vector %(of inferred size) 
whose $i$th coordinate, denoted by $e_{i,i}$ is 1, and 0 otherwise.

In this work, we will often talk about {\em $L_1$-normalized} vectors, by which we mean vectors that have been scaled such that their $L_1$-norm---sum of all entries---sum up to 1. Specifically, let $\norm{1}{\cdot}$ be the $L_1$ norm. We define an $L_1$ normalizing function, $\sf{N}\colon \inb{0,1}^n \to \inb{0,1}^n$ such that,
\begin{equation*}
    \sf{N}(\vv{v}) = \inp{\frac{v_1}{\norm{1}{\vv{v}}}, \ldots, \frac{v_m}{\norm{1}{\vv{v}}}}
\end{equation*}

We call $\sf{N}(\vv{v})$ the $L_1$-normalized vector of $\vv{v}$ and denote its $j$th coordinate by $\sf{N}(\vv{v})_j$.

\subsection{PageRank}
\label{sec:pr}

We present the \emph{weighted} graph version of PageRank \cite{pagerank,1344743} which readily generalizes the unweighted link graph structure in which it was initially proposed. Intuitively,  one can view weighted PageRank (on a weighted graph) as standard PageRank where the weights are captured by adding more links (proportionally to the weights). In applications, such weights correspond to some  real number  indicating ``trust'', ``importance'' or quality of the outgoing link. We assume weights are non-negative.

Let $G = (V,E)$ be a weighted, directed graph.
Let $\abs{V} = n$. Write $V = \inp{v_1, \ldots v_n}$ with respect to an appropriate indexing.
We have $\inp{u,v,w} \in E$ if there exists an edge from $u$ to $v$ with (non-negative) weight, $w$.
Let $w_{\text{out}}(v)$ be the sum of the weight of all outgoing edges of $v\in V$. We assume $w_{\text{out}}(v) > 0$ for all $v\in V$ since PageRank is not well-defined for ``dangling nodes''.
Let $w_{\rm out}(V)$ equal the vector of $w_{out}(v_i)$ for all $v_i\in V$, and $W_{\rm out}$ be the $n\times n$ diagonal matrix with diagonal equal to $w_{\rm out}(V)$.
Let $M$ be the adjacency matrix of $G$, with respect to the same indexing of $V$.

Fix \emph{restart probability} constant, $\alpha\in\inp{0,1}$. The PageRank vector, $\vv{\pi}$ is the solution to the following equation,
\begin{equation*}
    \vv{\pi} = \vv{\pi}\inp{1-\alpha}W_{\rm out}^{-1}M + \dfrac{\alpha}{n}\cdot\vv{1} 
\end{equation*}
where $\vv{1}$ is the vector of all 1's and with the constraint that $\norm{1}{\vv{\pi}} = 1$.
We overload noation and denote $\vv{\pi}(v_i) = \pi_i$ as the {\em PageRank}, aka the {\em rank}, of $v_i$.

We can also model PageRank as the stationary distribution of a row-stochastic markov chain, 
\begin{equation*}
    M' = \inp{1-\alpha}W_{\rm out}^{-1} M + \dfrac{\alpha}{n}\cdot \vv{1}_{n\times n}
\end{equation*}

i.e., the PageRank vector is the solution to,
 \begin{equation*}
    \vv{\pi} = \vv{\pi}M'
\end{equation*}

Observe that $W_{\rm out}^{-1}M$ is the $L_1$-row-normalized adjacency matrix of $G$ and can be viewed as the transition matrix of $G$ where the probability of transition is directly proportional to the relative weight on the outgoing edge. The row-stochastic matrix $\frac{1}{n}\mathbf{1_{n\times n}}$ captures the idea of {\em restarting} uniformly to any vertex in the graph. The convex combination of row-stochastic matrices is also row-stochastic and therefore, $M'$ is row-stochastic.

When are we guaranteed that the stationary distribution (the PageRank vector) $\vv{\pi}$ exists?
We present the following lemma whose proof can be found in most standard  elementary textbook on stochastic processes: 
\begin{lemma}[\cite{Haggstrom_2002,LevinPeresWilmer2006}]
    A finite markov chain has a {\em unique} stationary distribution if its transition matrix, $M$ is {\em irreducible}. If the chain is also {\em aperiodic} then the limiting distribution converges to the stationary distribution.
\end{lemma}

It is not hard to see that $M'$ is irreducible and aperiodic for any $G$ (due to {\em restart}) and therefore, has a unique stationary distribution invariant of what distribution you start with. (In this work, we are primarily interested in the existence and uniqueness of the stationary distribution for which irreducibility is sufficient.) Therefore, the PageRank vector is well-defined.

Later, it will be helpful to examine PageRank as a system of linear equations. In particular, we can rewrite the above equation as,
\begin{equation*}
    \vv{\pi}(v_i) = \sum_{j\in[n]}{\vv{\pi}(v_j)\cdot M'[i,j]}
\end{equation*}

% We know \emph{irreducible} finite markov chain has a unique stationary distribution. We can rewrite the above equation as follows,
% \begin{equation*}
%     \vv{\pi} = \vv{\pi}\inb{\inp{1-\alpha}M + \dfrac{\alpha}{n}\cdot \vv{1}_{n\times n}}
% \end{equation*}
% where $\vv{1}_{n\times n}$ is the $n\times n$ matrix of all 1's. Observe that $\inp{1-\alpha}M + \dfrac{\alpha}{n}\cdot \vv{1}_{n\times n}$ is both row-stochastic and irreducible, and therefore $\vv{\pi}$ exists and is unique.

%
\subsection{Personalized PageRank and Contribution}
\label{sec:ppr}

PageRank is an emergent global property of the graph. After taking in all edges into consideration, it calculates a global ``ranking'' or importance score for each vertex.
{\em Personalized} PageRank \cite{pagerank} aims to capture the (bidirectional) relationship between any two vertices in the graph; it can be seen as the ``significance'' of some  target node, $t$ with respect to some source node, $s$ or also, the ``importance'' of $s$ from the perspective of $t$ \cite{10.1145/1052934.1052938,10.1007/978-3-540-77004-6_12,10471277,10.1145/2488388.2488433,10.1093/bioinformatics/btq680}. In this work, we will be more interested in the latter interpretation as we will measure the importance of users in the eventual ranking of servers and distribute rewards proportional to the relative importance (which we later call, {\em contribution}) of each user.

Formally, the Personalized PageRank vector for vertex $s$, $\vv{\pi_s} = \inp{\pi_{s,1}, \ldots, \pi_{s,n}}$ assigns a numeric {\em score} to each $v_i$ in the graph. Let $s = v_k$. We overload notation again and denote, $\vv{\pi_s}(v_i) = \pi_{s,i}$. For the same restart-probability constant as in PageRank, $\vv{\pi_s}$ is the (unique) solution to the following equation,
\[%\begin{equation*}
    \vv{\pi_s} = \vv{\pi_s}(1-\alpha)W_{\rm out}^{-1}M + \alpha \vv{e_k}
%\end{equation*}
\textrm{, with } \norm{1}{\vv{\pi_s}} = 1.\]

Personalized PageRank can be interpreted as a random walk over the graph that starts at $s$ and with probability $1-\alpha$ travels to a neighboring vertex along the path of the random walk,  %with probability specified by $W_{\rm out}^{-1}M$, 
and with probability $\alpha$ restarts at $s$.
Then, the significance of $t$ with respect to $s$ is how frequently the random walk passes through the vertex, $t$ and is equal to, $\vv{\pi_s}(t)$. In other words, we can also model Personalized PageRank as the stationary distribution of a row-stochastic markov chain,
\begin{equation*}
    M'_s = (1-\alpha)W_{\rm out}^{-1}M +\alpha E_s
\end{equation*}
where $E_s$ is the $n\times n$ matrix in which every row is equal to $\vv{e_s}$.

As mentioned, Personalized PageRank can also be interpreted from the view of $t$. Define the {\em inverse} Personalized PageRank vector, aka  {\em Contribution} PageRank vector, for vertex $t$, $\vv{\pi^{-1}_t}$ as the vector of $\vv{\pi_{s}}(t)$ for all $s\in V$. In other words, $\vv{\pi^{-1}_t}(s) = \vv{\pi_s}(t)$. Since 
$\norm{1}{\vv{\pi_s}}$, we have $\sum_{t\in V}\vv{\pi^{-1}_t}(s) = 1$.

We define, the {\em relative contribution} PageRank vector for $t$ as, 
\begin{equation*}
    \vv{\omega_t^{-1}}(s) = \dfrac{\vv{\pi^{-1}_t}(s)}{\sum_{s\in V} \vv{\pi^{-1}_t}(s)}
\end{equation*}

% We state without proof that $\sum_{s\in V} \vv{\pi^{-1}_t}(s) = \vv{\pi}(t)$, the PageRank of $t$ is the sum of all its contributions.
% Observe that, $\vv{\omega_t^{-1}} = \sf{N}\inp{\vv{\pi^{-1}_t}}$.

Personalized PageRank can also be generalized to any distribution of source nodes by changing $\vv{e_k}$ in the equation. For example, we can measure the significance of $t$ with respect to a set of vertices, $S$, of size $m$ instead of a single vertex by replacing $\vv{e}_s$ with a vector that is equal to $1/m$ at each coordinate that corresponds to $S$. Intuitively, this corresponds to the random walk restarting from any vertex in $S$ uniformly at random. Then, Personalized PageRank of $t$ with respect to $S$ is how frequently a random walk starting (and restarting) from $S$ passes through $t$. 

% end cancel

%
% \subsection{Graph Model}
% \rohit{maybe this submission doesn't get into layered graphs and dynamic participation??}
% \begin{definition}[Layered Graph]
%     Let $G = (V,E)$ be a directed, weighted graph. We say $G$ is \emph{$k$-layered} for some $k \in \mathbb{Z}^{+}$ if there exists a $k$-partition of $V$, $\inset{V_1, V_2, \ldots, V_k}$ such that for all $\inp{u, v, w} \in E$ if $u \in V_i$ then $v \in V_{i-1}$. In other words, the graph can be split into $k$ ``layers'' such that edges exist only between consecutive layers (in one direction). We denote $V_i$ as layer $i$.
% \end{definition}

\cancel{
\subsection{Fair Lotteries and Proof-of-Reputation Blockchain cite[KOZ20]}

\begin{itemize}
    \item blockchain and consensus
    \item proof of reputation
    \item correlation-free static reputation system
    \item implementing fair lottery
\end{itemize}
}

 \section{Technical Overview}\label{sec:tech_overview}

We start by discussing how PageRank can be tuned %used (and adapted)
to address the needs of Problem 1 above. %---distilling a meaningful reputation system.   
To provide some background for readers less familiar with PageRank: the idea of PageRank is to model the Internet as a graph, where the nodes are pages and a link in page $P$ pointing to page $P'$ corresponds to a directed edge in this graph from $P$ to $P'$. The algorithm can be seen as a random walk on this graph, which starts at a random page, and in every step one might either move along one of the outgoing edges of the current node or, with some given probability, ``teleport'' to some other node. Interestingly, this last idea was originally introduced as a technical ``hack'' to ensure that the walk does not get ``stuck'' on a sink---node without outgoing edges---or a {\em terminal} node---the only outgoing edge being a self-loop. But as we will discuss here, in solving our problem, teleportation is relevant even in graphs without sinks or terminals, and has an interesting interpretation that relates to how well-informed a node is. 

The (Page)Rank of any node is computed based on the frequency that the node is visited in comparison to the overall time of the walk. %The details of the algorithm are not necessary for understanding the intuition of our solution. 
In the Preliminaries (Section~\ref{sec:prelim}), we offer the details one needs to understand  the parameters of our solution, and how those are derived from PageRank. {\em Weighted PageRank}~\cite{1344743} is the straightforward extension of PageRank to a weighted graph, where an edge leaving a node $P$ is followed with probability proportional to its relative weight compared to the weights of other edges in $P$'s outgoing neighborhood.  

A {\em reputation graph}, encodes  the blockchain nodes' beliefs analogously to how the Web graph encodes the Internet: The vertices are the blockchain nodes, and an edge from a node $P_i$ to a node $P_j$ is weighted by  $P_i$'s belief about $P_j$'s trustworthiness: i.e., and edge $(P_i,P_j)$ with weight $w_{ij}\in[0,1]$ encodes $P_i$'s belief that the probability of $P_j$ following its specification is $w_{i,j}$. Given the above we investigate the question of whether PageRank demonstrates the anticipated behavior.

\paragraph{\textbf{Vanilla PageRank does not work.}}
\label{sec:vanillaPR} 
The first attempt is to consider standard PageRank. Here, we use the simplest possible setting as a benchmark, where every node is perfectly informed about the trustworthiness $\sf{R}_j$ of every node $s_j$ in the system. That is, the reputation graph is a clique with $m$ nodes (corresponding to $m$ blockchain nodes), where for each edge $(s_i,s_j), i,j\in[m]$ its associated weight is $w_{ij}=\sf{R}_j$ (normalized so that the total weight leaving each node is $1$.) 
%. \vassilis{the $R_j$ don't sum up to 1 as is. we need to normalize.}

In fact running PageRank on this graph yields, for each blockchain node,  $s_i$,  a rank, which we call {\em reputation score}, such that the vector of {\em computed reputation scores} is ratio-preserving (and therefore also order-preserving) with respect to the vector of the actual nodes' trustworthiness scores. (See Lemma~\ref{lem:pr-on-clique}.) That is, if nodes $s_i$ and $s_j$ have trustworthiness scores $\vrep_{i}$ and $\vrep_{j}$, respectively, then for the corresponding {\em reputation scores} $\vrsc_{i}$ and $\vrsc_{j}$, as computed by running PageRank on our encoding/graph it holds that $\vrsc_{i}/\vrsc_{j}=\vrep_{i}/\vrep_{j}$. This implies that ordering the nodes according to their computed reputation score is the same as ordering them according to their trustworthiness score (ground truth). We note in passing that such an order-preserving ratio is the best we could hope to achieve from PageRank; this is because in order to use PageRank we need to make the graph's adjacency matrix (row) stochastic---each row must sum up to 1---which inherently normalizes the nodes' beliefs. However, deriving such a ratio-preserving representation is sufficient for most if not all applications of reputation systems, where one primarily cares about the order and ratio between ranks. %Indeed, one can easily scale this up to derive an x-out-of-5$ 

The above appears as a validation of using PageRank in our scenario; but it can only treat the scenario with very strong assumptions of perfectly informed parties. In fact, it is not hard to see that even a single node being misinformed (or uninformed) about the trustworthiness of some other node can completely throw off the output of the algorithm.
\cancel{
consider the case in which the most influential node is the uninformed one...
}

\paragraph{\textbf{Designated PageRank for a subset.}}
The above negative result motivates the question: Are there interesting scenarios where PageRank does not demonstrate the above counterintuitive and undesired behavior? To answer this question we turn back to our original motivation: In several distributed systems scenarios, including blockchain systems, the nodes have distinct roles (e.g., blockchain users, proposers, validators, stake operators). As such only a subset of nodes might need to be ranked for each task. Hence one can ask whether or not by limiting the  meaningfulness requirement to only a specific subset we can get graph topologies, which make PageRank more robust to uninformed nodes. 

To this direction, consider a blockchain with $n$ nodes from a set $V$, where we are trying to rank a subset $\{v_{\ell_1},\ldots,v_{\ell_m}\}\subseteq V$ of these nodes, which we will refer to as the {\em (designated) reputation nodes} or simply as {\em servers.} In particular, we augment the original weighted ($n$-node) graph $G=(V,E)$ discussed above by adding $m$ sink nodes $\{s_1,\ldots,s_m\}=S$ that, intuitively are replicas of the nodes we wish to rank without their outgoing edges. In particular, the new graph consists of the $n+m$ nodes in $V\cup S$ and includes all edges in the original graph $G$ with the same weight, along with $n m$ new edges from each node $v_i\in V$ to each node $s_j\in S$ such that every such edge $(v_i,s_j)$, which we will refer to as {\em $v_i$'s endorsement of $s_j$}, has weight $\bar{w}_{i,j}=w_{i,\ell_j}$. I.e., the weight of each endorsement of $(v_i,s_j)$ any node $s_j\in S$ is the same as the weight of the edge to $s_j$'s replica in $G$ (i.e.,   $v_{\ell_j}$). 

Similar to our vanilla PageRank application, to run PageRank on this graph we will need to ensure that the adjacency matrix of this new graph is stochastic. For non-sinks, this can be done by normalizing the weights of their outgoing neighborhood as in the vanilla application of PageRank. However, for sinks we need to be careful: PageRank itself has an extension to handle sinks:\footnote{Sinks are problematic in PageRank as the associated random walk will get stuck if it reaches a sink.}
One can change the graph to add $n$ edges with equal weights from the sink to every node in the graph. This fix works for applications of PageRank, but fails in our asymmetric setting where nodes have different roles. Intuitively the reason is that introducing edges between nodes in $S$ which have uniform (in $S$) weights, bias the random walk towards them in a way that does not depend on the system's belief, which, in turns, steers PageRank away from these beliefs. Instead, we will change the above to have sinks point to every $V$-node with the same weight (but not to any other $S$ node). This will ensure that the random walk, and hence also the resulting {\em rank} of each $s_i$ is not biased by this modification.  In fact, for a similar reason, we need to ensure that PageRank's teleportation step might only land on a $V$ node---i.e., it never lands on a server. 

For clarity, we will refer to the above graph as {\em (Designated) Reputation Graph} and to the version of the PageRank algorithm applied to the above graph with the above teleportation modification as {\em Designated PageRank}. Our next technical contributions are  proving that Designated PageRank improves on the properties of PageRank for our (designated reputation nodes) scenario.

\paragraph{Perfect Information.} As a sanity check, we confirm  that if $G$ is the perfectly informed clique from above, then Designated PageRank produces a ratio-preserving (computed) reputation score/vector for the designated reputation nodes. 
We elaborate in Lemma~\ref{lem:pr-on-moon-clique}. 
%which states that as long as the nodes are perfectly informed about just the reputation nodes (even if they are fully incorrect about all other nodes), Designated PageRank will still produce a ratio-preserving reputation score. 

\paragraph{Perfect Information for only a subset.} Next, motivated by the fact that in the blockchain setting, there are natural situations in which a subset $V'\subset V$ of nodes is better informed about the trustworthiness of certain nodes than other---e.g., because they have been longer with the system---we consider a further fine-tuned version of Designated PageRank; here  all $V\setminus V'$-to-S edges are removed (and the associated weight are re-normalized. Also in this setting, we prove that as long as the nodes in $V'$ are perfectly informed, Designated PageRank on this new graph still yields an order-preserving reputation vector irrespectively of how uninformed nodes is  $V\setminus V'$ are, or how uninformed nodes in $V'$ are about nodes in  $V\setminus V'$. See Section~\ref{sec:hierarchical-info}.

\paragraph{Imperfect (consistently noisy) information.} Finally we investigate another relaxation on how informed nodes are, which is motivated by a more realistic blockchain scenario. We assume that all players have a noisy but  {\em consistent} belief about the reputation nodes' trustworthiness,  i.e., they all know a value $\vtype_j\in[0,1]$
(for each $j\in[\vns]$) where with high probability (confidence), $\vtype_j$ is within some $\epsilon$ from the the trustworthiness score of server $j$. This is a natural scenario which can, for example, occur when all nodes apply a consistent statistic on public data about each reputation nodes. Here we can prove that we can compute a good approximation of an order preserving reputation system. In particular, running Designated PageRank we can compute a reputation score $\vrsc_i$ for each reputation node $i\in[\vns]$,  so that  for any $i, j\in[m]$, with high probability, $\vrsc_{i}/\vrsc_{j}$ is within an $\epsilon''$  (diminishing in $m$) factor from $\vrep_{i}/\vrep_j$.   Under assumptions about the density of the trustworthiness score vector and/or the number of servers, the above allows to limit with high probability the number of {\em inversions} in the ordering derived by the decoded reputation scores. We present this case in \Cref{sec:imperf-id-eps}.

\begin{remark}
It is worth noting that although the above treatment is applied to rank nodes in the original node-set $V$ of $G$, all our results can be applied to score a set of nodes $V'$ disjoint with $V$, as long as the corresponding assumptions on the beliefs of parties in $V$ about trustworthiness of $V'$-nodes holds. 
As such, we believe that our treatment can be applied to allow scoring also events that are exogenous to the blockchain, as long that the nodes can consistently observe them. We find this an exciting future direction. 

In Section~\ref{sec:trep}, we formally present a generalization of Designated Reputation Graphs that we call \emph{Trustworthy Reputation (TRep) Graphs} that generalizes the concept of replicas to an arbitrary subset of nodes. This graph is in correspondence with and complements the game we define below. We prove all our results about Designated PageRank and Reputation Graphs in this generalized setting. 
\end{remark}

\subsection{Trustworthy Reputation Games} 
We next discuss our techniques for solving Problem~2 discussed in Section~\ref{sec:contributions}.

Our goal  is to design games that allow the blockchain to extract an order-preserving estimate of reputation from its users, assuming they are rational. Towards this goal, we  define  a class of games, which we term {\em trustworthy reputation games} (in short, {\em TRep games}) that capture the intuition of the above goal. In a nutshell, a {\em trustworthy reputation game} is an $\vnc$-agent Bayesian game against nature, where the (private) state of nature is encoded as a vector of $\vns$ values (probabilities) in $[0,1]$, ${\vvrep}=(\vrep_1,\ldots, \vrep_\vns)$.

%In our reputation-based blockchain scenario above, $\vns$ corresponds to the number of blockchain servers/nodes and each $\vrep_j$  correspond to the ground truth (trustworthiness score) of the $j$th blockchain node. Note that in this work we  focus on blochchains with a  fixed universe of $m$ nodes. % and blockchain users. 
%Extending our treatment to dynamically changing universe is an interesting future direction.  Notwithstanding, to allow most generality we allow these sets to be disjoint.

Let us use as motivation the task of the blockchain selecting a (small) committee of its nodes according to their reputation. We will refer to the nodes that might be selected as {\em servers,} and to the complete node set as {\em users} (using the notation from the previous section, $V$ is the set of users, and $S$ is the set of servers.) To form such a committee,  (anyone with read-access to) the blockchain needs to  \emph{discover} $\vv{\sf{R}}$ by observing messages posted by its users.

The set of pure actions of each user (node in $V$) in a TRep game consists of endorsements to one of the reputation nodes in $S$ or to other users; hence a mixed strategy (which is what the user should record on the blockchain to receive associated rewards) is a probability distribution over the set of the  $m$ servers and $n$ users. A key feature of trustworthy reputation games is that they come equipped with an efficient {\em decoding} function, which, given any strategy profile from a given class, computes a specified {\em reputation} function of the ground truth. Looking ahead, in the games we design, this class will consist of   Nash equilibria (NE) and the reputation function will provide an estimate %outputs a vector which approximates 
of the \emph{relative} trustworthiness between any two servers, 
$\vrep_{i}/\vrep_{j}$ for $i,j\in[\vns]$;  such an estimate  will allow us to order the servers according to their trustworthiness, which can then be used by a  blockchain to solve the above reputation-based  committee selection problem. 
As a definition contribution,  we provide a formal definition of TRep games in Section~\ref{sec:trep}. 

The next technical challenge is instantiating the TRep game with a utility function which (1) ensures that NEs can be decoded as above, and (2) can be realized using blockchain rewards. The intuition of our choice is that in order to incentivize users to endorse in a way consistent to their beliefs, we need to ensure that any rewards that are collected as a result of their endorsement are also distributed equitably according to the nodes relevance. This points to the following choice: Assume that nature samples a binary vector in the support of ${\vvrep}$, where the $i$th component is $1$ with probability $\vrep_i$; then if the $i$-th componnent is $1$ give a fixed rewards to $s_i$, and diffuse an equal amount to this reward to users (in $V$) according to how much their endorsements helped approximate the trustworthiness of the nodes; in other words, each user's reward is proportional to their contribution on the server's computed reputation. 

Interestingly, the PageRank literature has a notions which is aimed at exactly computing such a contribution. This is the so-called {\em Personalized } PageRank~\cite{10.1145/1052934.1052938,10.1007/978-3-540-77004-6_12,10471277,10.1145/2488388.2488433,10.1093/bioinformatics/btq680}, which, intuitively, computes the importance of every node $u$ relative to a specific node $v$, i.e., how relevant the endorsement of $u$ was in $v$ receiving its PageRank.   As we show, by instantiating the diffusion function in the above reward mechanism using the output of Personalized PageRank on the Designated Reputation Graphs from the previous section, we obtain a utility function which satisfies all the properties of TRep games. In particular we analyze all three scenarios discussed and prove the following stability results: 

\paragraph{Perfect Information.} Assume that   every user knows the vector $\vec{\vrep}=(\vrep_1,\ldots,\vrep_\vns)$ (This degenerate game instance makes the TRep game a complete information game against nature).  We prove the following result for this game, $\vgame_{\rm perfect}$: 

\begin{theorem}[informal]
 $\vgame_{\rm perfect}$ has a unique Nash equilibrium (NE), where by observing the players' NE strategies, we can compute a reputation score $\vrsc_i$ for each server $i\in[\vns]$, such that for any $i, j\in[m], \mid \vrsc_{i}/\vrsc_{j}=\vrep_{i}/\vrep_j$.  
\end{theorem}

Interestingly, the NE strategy in the above game is not reporting all the weights, but rather just the ones pointing to the servers in $S$. Notwithstanding,  as we prove, running Designated PageRank on this subgraph induces (on $S$) the same ratio-preserving reputation system as running it on the full graph. In fact, a similar effect is observed also in the other two cases which follow. In our view, this phenomenon attest to how  non-trivial it is to chose the right instantiation for a utility function. 

\paragraph{Perfect Information for only a subset.} The above statement holds even in this case. This is stated and proved in Section~\ref{sec:hierarchical-info}.

\paragraph{Imperfect (consistently noisy) information.}  Finally we look at the game where every user has a noisy but  {\em consistent} view of nature's state, i.e., they all know a value $\vtype_j\in[0,1]$
(for each $j\in[\vns]$) where with high probability (confidence), $\vtype_j$ is within some $\epsilon$ from the the trustworthiness score of server $j$. % We consider this as a natural scenario which can, for example, occur by all parties applying a consistent statistic on public data about each server. 
We prove the following result for this game, $\vgame_{\rm noisy}$ in \Cref{sec:imperf-id-eps}.

\begin{theorem}[informal]
In $\vgame_{\rm noisy}$,  playing according to their  beliefs/types is $\epsilon'$-Nash for an $\epsilon'$ that diminishes with the size of the user set. By observing the players'  strategies, we can compute a reputation score $\vrsc_i$ for each server $i\in[\vns]$,  so that such that for any $i, j\in[m]$, with high probability, $\vrsc_{i}/\vrsc_{j}$ is within an $\epsilon''$  (diminishing in $m$) factor from $\vrep_{i}/\vrep_j$.  
\end{theorem}

Under assumptions about the density of the trustworthiness score vector and/or the number of servers, the above allows to limit with high probability the number of {\em inversions} in the ordering derived by the decoded reputation scores. 

\subsection{Application to PoR/PoS-Blockchains.}
Finally, in Appendix~\ref{sec:PoR/PoS} we discuss how the above games can be used within a reputation-based blockchain. The idea is that we can associate nature's state (i.e., trustworthiness score of each of each nodes/servers) with the probability that that the node follows the protocol. The blockchain, then, mints for each server who followed the protocol a fixed amount of coins which is distributed according to our (Personalized PageRank-based) utility function. For clarity, we focus our treatment on the hybrid Proof-of-Reputation/Proof-of-Stake blockcain by~\cite{KleOstZik20}. However, similar arguments can also be applied to other systems that rely on reputation-based lotteries. 

\section{Trustworthy Reputation Games}
\label{sec:trep}
We formally define a new class of games called {\em Trustworthy Reputation Games}. We show how to use PageRank to define meaningful utilities, and a ``meta-objective'' of the game, which we formalize using the notion of {\em decodability}, which we introduce to our games.  
% the game being played and the decodability property. 
% Later, we represent it abstractly as a bi-partite graph to model client-server networks and specifically blockchain-related settings.
%
As a reminder, the  goal of the model is for users/agents to act in accordance to their beliefs about the trustworthiness of the servers and/or other agents. As a game, we model the trustworthiness scores of the servers as {\em nature's} (private) state and the belief as a {\em type}. We thus, model our game as a Bayesian game~\cite{Harsanyi1982,Zamir2009} where nature assigns a type to each player representing each players' belief about nature's private state before the game begins.

\begin{definition}[Trustworthy Reputation Game with $\inp{\cal{E},f}$-Decodability]
Let $n, m \in \mathbb{Z}_{\geq 2}$.
 A {\em Trustworthy Reputation Game} or {\em TRep Game} is a {\em simultaneous} Bayesian game {\em against nature} and is defined as a tuple,\\ $\cal{G} = \inp{\cal{P}, \cal{A} = \prod_{i\in[n]}{\cal{A}_i}, \inp{u_i}_{i\in[n]}, \inp{T_i}_{i\in[n]}, \inp{\sf{R}_j}_{j\in[m]}}$ where,
\begin{enumerate}[label=(\alph*)]
    \item \textbf{Players.} $\cal{P} = \inp{P_i}_{i\in[n]}$ is the set of $n$ agents/players,
    \item \textbf{Action Space.} $\cal{A}$ is the action space (or equivalently pure strategy space since the game is simultaneous) of all the players. Every player has the same set of actions, $\cal{A}_i = \inb{m+n}$, and so $\cal{A} = \inb{m+n}^n$. Let $a_i = j$ if $P_i$ picks $j$ as its action. Let $\Delta_i$ be the set of all probability distributions over $\cal{A}_i$, i.e., the set of all (mixed) strategies of $P_i$. Strategies are represented as an $(m+n)$-vector of probabilities. Let $\Delta = \prod_{i\in[n]}\Delta_i$, the set of all strategy profiles of the game.
    \item \textbf{Nature.} $\inp{\sf{R}_j}_{j\in[m]}$ is nature's (private) state. Each $\sf{R}_j\in\inb{0,1}$ is interpreted as a probability. We treat Nature as a non-strategic player that always plays the same strategy and has no payoff from the game. Nature's move is determined by flipping $m$ biased coins or Bernoulli random variables, $H_j\sim {\rm Ber}(\sf{R}_j)$, i.e., $\Pr\inb{H_j = 1} = \sf{R}_j$. Denote by $\cal{N} = \inset{0,1}^m$, the space of nature's move, and $\Delta_{\cal{N}}$, the set of all probability distributions over $\cal{N}$. So, nature will always play the (mixed) strategy, $\vv{\sf{R}} = \inp{\sf{R}_1, \ldots, \sf{R}_m} \in \Delta_{\cal{N}}$.

    \item \textbf{Types.}   Players may not have perfect information about nature's state which we capture by assigning a ``type'' to each player. 
    Let $T_i$ be the type space of $P_i$. $\vv{t_i} \in T_i$ will represent $P_i$'s belief about nature's state. (In the following sections, we will give instantiations of the general game by considering different type spaces.)
     Denote by $\mathcal{T} = \prod_{i\in[n]} T_i$, the space of all players' types.

     % \item \textbf{Common prior.} $p\colon \cal{T} \to \inb{0,1}$ is a probability distribution over the type space. We assume that every player knows $p$.
    
    \item \textbf{Utilities.} $u_i\colon \cal{T} \times \cal{A} \times \cal{N} \to \inb{0,1}$ is the utility of $P_i$ and depends on the players' types, actions and nature's move. When referring to the (expected) utility with respected to mixed strategies, we overload notation as is standard and write, $u_i\colon \cal{T}\times\Delta\times \Delta_{\cal{N}} \to\inb{0,1}$. We formally define the utilities in the next section.
    % Since we are looking at Bayesian games where the players may be uncertain about nature's state, the utility can be interpreted as their (perceived) payoff given their belief about nature (their type).
    %For the games we present in this work, $u_i$ will only depend on the type of $P_i$ and not others so we omit taking as input the other players types and equivalently write, $u_i\colon T_i \times \cal{A}$.

    \item \textbf{Decodability.} Let $\cal{E}\subseteq \Delta$ be some set of possible strategy profiles of $\cal{G}$, and let\linebreak $f~\colon~\inb{0,1}^m~\to~ \inb{0,1}^*$ be a possibly randomized function with codomain as vectors of arbitrary length over $\inb{0,1}$.
 We say that a TRep game, $\cal{G}$ is {\em $\inp{\cal{E},f}$-decodable} if there exists an {\em efficient} decoding function, $\cal{D}\colon \Delta \to \inb{0,1}^*$ such that when sampled using any $\vv{e}\in\cal{E}$, is identical to $f$ sampled using nature's private state, i.e.,
    \begin{equation*}
        {\cal{D}(\vv{e}) \simeq f\inp{\sf{R}_1, \ldots, \sf{R}_m}}
    \end{equation*}
   
    We call $f$ the {\em reputation} function, and we say, the set of strategy profiles, $\cal{E}$, $f$-{\em encodes} nature's private state.

%     The game proceeds as follows,
%    \begin{enumerate}[label=(\arabic*)]
%     \item Nature assigns each player, $P_i$, a type $\vv{t_i}\in T_i$,
    
%     \item Simultaneously, each $P_i \in \cal{P}$ plays $a_i \in \cal{A}$, and nature plays its move accordingly.
% \end{enumerate}

\end{enumerate}

\end{definition}

In words, the idea of the game is for the players' to encode some function over nature's private state with a suitable set of strategy profiles. They may have some belief about nature's state {\em a priori} which influences their utility and strategies.

Before we introduce the utilities, we translate the structure of our game into a graph which formalizes the designated model we are interested in.

\begin{definition}[Trustworthy Reputation Graphs]
    \\Let $G=\inp{V=\cal{V}\sqcup\hat{\cal{V}}, E, \vv{\sf{R}}}$ be a directed, edge-weighted, partially vertex-weighted graph, where $\vv{\sf{R}}$ is the vector of vertex weights on $\hat{\cal{V}}$ only, and $\inp{u, v, w} \in E$ iff there exists an edge from $u\in{\cal{V}}$ to $v \in V$ with weight $w$. We denote by $\cal{V} = \inset{v_1, \ldots, v_n}$  the set of $n$ {\em users} and by $\hat{\cal{V}} = \inset{\hat{v}_1, \ldots, \hat{v}_m}$  the set of $m$ {\em servers}.
 
    $G$ is a \emph{Trustworthy Reputation Graph} or \emph{TRep Graph} if,
    \begin{enumerate}[label=(\alph*)]
        \item vertices in $\hat{\cal{V}}$ have out-degree 0, 
        % \rohit{at this point it should already makes sense why we don't allow outgoing edges from $\hat{\cal{V}}$ and it's because we're interested in their rank as calculated by everyone else so they should not be allowed to affect the system.}
        \item weights are non-negative and for all $u\in\cal{V}$,  $\sum_{\inp{u,\cdot,w}}w = 1$, i.e., the weights on the outgoing edges from each vertex, respectively, sum up to 1,
        \item $\vv{\sf{R}} = \inp{\sf{R}_1, \ldots, \sf{R}_m}$, and $\sf{R}_j\in\inb{0,1}$ for all $j\in[m]$, i.e., the weight on each vertex in $\hat{\cal{V}}$ is in $\inb{0,1}$.
    \end{enumerate}
    
        For each $v_i\in\cal{V}$, let $\vv{w^{(i)}}$ be the vector of edge weights. So $w^{(i)}_j $ is equal to $w$ if $\inp{v_i, \hat{v}_j, w} \in E$ and 0 otherwise. We call this vector, $\vv{w^{(i)}}$ the \emph{endorsements} of $v_i$.

        For each $\hat{v}_j\in\hat{\cal{V}}$, denote $\sf{R}_j$ the \emph{trustworthiness score} of $\hat{v}_j$. 
        %The function $f$ translates $\hat{\cal{V}}$'s trustworthiness scores to reputation scores.
\end{definition}

In effect, the edge weights from each vertex in $\cal{V}$ are modeled as a probability distribution over $V$.

It is easy to see the correspondence between TRep games and TRep graphs. The mixed strategies of each player correspond to the outgoing weighted edges from $\cal{V}$ (with respect to some indexing), and nature's state corresponds to the vertex weights. In TRep games, $a_i = j$ for $j\in\inset{1,\ldots, m}$ is $P_i$ \emph{endorsing}  server $j$, and for $j\in\inset{m+1, \ldots m+n}$ is $P_i$ endorsing player $P_{j-m}$.

% \rohit{talk about the adjacency matrix or the random surfer model with reset to everyone uniformly after reaching a sink.}

We can now %are now ready to 
introduce the PageRank-inspired utility that underlies TRep games.

\subsection{Defining Utilities and Decoding Using PageRank}
\label{sec:util-pr}
We examine the (weighted) PageRank algorithm on TRep graphs.
First, we show how PageRank can be used to evaluate servers which we later use to derive the decoder. Next, we show a corresponding utility function,  defined using Personalized PageRank, that when used in TRep games enable PageRank's evaluation to recover (i.e., decode) the (relative) trustworthiness of the servers.

\paragraph{\textbf{Evaluating Servers.}}
Let $G$ be a TRep graph.
We cannot use such a graph structure readily with PageRank due to the issue of {\em dangling} vertices. Observe that every vertex in $\hat{\cal{V}}$ is dangling (out-degree 0). 
There are several folklore fixes. One fix in particular is to transition to any vertex with equal probability. In effect, the random walk \emph{resets} (with probability 1) every time it encounters a dangling vertex or ``sink''. We build on this fix and make the following modification: instead of teleporting to \emph{any} vertex, we restrict teleportation to vertices in $\cal{V}$ with equal probability. This modification is not arbitrary and is in fact reminiscent of Personalized PageRank---in our model, we are only interested in the importance of the servers, $\hat{\cal{V}}$, derived from the users, $\hat{\cal{V}}$; more specifically, we are interested in the relative contribution PageRank of the users for each server. Thus, we are not interested in the contribution of $\hat{\cal{V}}$ to itself and each other. 

To this extent, we restrict teleportation (i.e., when the random surfer decides to reset) to only $\cal{V}$ uniformly at random.
%(the effect of this is equivalent to adding uniform-weight edges from each server to all users).

% We make a modification to the general PageRank algorithm reminiscent of Personalized PageRank. In the case where the random surfer decides to restart, instead of teleporting to any vertex uniformly at random, we restrict its teleportation to $\cal{V}$ uniformly at random. 
% That is, a random walk will always restart from $\cal{V}$. The intuition behind this modification is analogous to (generalized) Personalized PageRank---in our model, we are only interested in the importance of the servers, $\hat{\cal{V}}$, derived from the users; more specifically, we are interested in the relative contribution PageRank of the users for each server.

Let $T$ be the transition matrix for $G$ with the above fix for sinks. We arrange the vertices such that the first $n$ indices correspond to $\cal{V}$ and the next $m$ correspond to $\hat{\cal{V}}$.
Then, we wish to calculate the PageRank $\vv{\pi}$ as the solution of the following,
\begin{equation*}
    \vv{\pi} = \vv{\pi}\inp{\inp{1-\alpha}{T}} + \alpha
    \begin{bmatrix}
        \frac{1}{n}\cdot \vv{1}_{n\times 1} &
        \vv{0}_{m\times 1}
    \end{bmatrix}
\end{equation*}

Let $T'$ be the Markov chain underlying the above relation (i.e., including restart).
We remark that $\vv{\pi}$ exists and is unique. Observe that clearly every user is {\em accessible} from any server. As for the other way around, there are two cases, (1) for all $v\in\hat{\cal{V}}$, there exists $u\in\cal{V}$ such that there is an edge from $u$ to $v$. In that case, the underlying Markov chain is clearly irreducible; (2) there exists $v\in\hat{\cal{V}}$, such that for all $u\in\cal{V}$, there is no edge from $u$ to $v$. In that case, $v$ is a {\em transient} state and its long term distribution approaches 0. In some sense, $v$ is irrelevant to the long term behavior of the Markov chain and can be ignored.
% \rohit{this proof even works when we just connect all servers to users!! no need for self loop alteration or anything... ugh how did I miss this...}
So long as the entire graph is not transient states, we can prune all such transient servers where their probability is 0 in the stationary distribution; the remainder of the states will constitute an irreducible Markov chain for which there exists a unique stationary distribution.
% (Moreover, the self-loops on $\hat{\cal{V}}$ also make $T'$ \emph{aperiodic} and thus the limiting distribution converges to the stationary distribution.)
% \rohit{turns out you just need irreducibility for uniqueness. aperiodic guarantees convergence but we don't care about that here...}
Note that this fits our intuition precisely---$v$'s {\em importance} as indicated by the stationary distribution is 0 as no user {\em trusts} it.

For the prescribed graph structure, we can calculate the PageRank of each server explicitly using the system of linear equations. Let $\Pr\inb{u\to v \mid T} = w^{(i)}_j$, the probability of transitioning from vertex $u$ to $v$ as specified in $T$.
Note that $\Pr[\hat{v}_i\to \hat{v}_j\mid T] = 0$ for any $i,j\in[m]$. Also, for $v\in V$, $\Pr[v\to \hat{v}_i\mid T'] = (1-\alpha) \Pr[v\to \hat{v}_i\mid T]$ as with probability $\alpha$ we restart to $\cal{V}$. Therefore, for each $\hat{v}\in\hat{\cal{V}}$, we have,
\begin{align*}
    \vv{\pi}(\hat{v})
    &= \sum_{\hat{u}\in{\hat{\cal{V}}}} \Pr\inb{\hat{u}\to \hat{v}\mid T'} \cdot \vv{\pi}(\hat{u}) + \sum_{v\in{\cal{V}}} \Pr\inb{v\to \hat{v}\mid T'} \cdot \vv{\pi}(v)
    \\
%    &= (1-\alpha)\Pr\inb{\hat{v}\to \hat{v}\mid T}\vv{\pi}(\hat{v}) + (1-\alpha)\sum_{v\in{\cal{V}}} \Pr\inb{v\to \hat{v}\mid T} \cdot \vv{\pi}(v)
 %   \\
    &= 0 + (1-\alpha)\sum_{v\in{\cal{V}}} \Pr\inb{v\to \hat{v}\mid T} \cdot \vv{\pi}(v)
    % \\
    % \implies
    % \vv{\pi}(\hat{v}) &= \frac{1-\alpha}{\alpha}\sum_{v\in{\cal{V}}} \Pr\inb{v\to \hat{v}\mid T} \cdot \vv{\pi}(v) \stepcounter{equation}\tag{\theequation}\label{eq:server-pr}
\end{align*}

Therefore, the PageRank of each server is directly proportional to the sum of its incoming {\em endorsements} from the users, $\cal{V}$ weighted by their PageRank. Since we are only interested in the PageRank of the servers, we normalize on the set of the servers and define, the ``reputation score'' of server $j$, $\hat{v}_j\in\hat{\cal{V}}$ as, 
\begin{equation}\label{eq:rep_score}
    \vrsc_j = \dfrac{\vv{\pi}(\hat{v}_j)}{\sum_{\hat{v}\in\hat{\cal{V}}}{\vv{\pi}(\hat{v})}} = \dfrac{ \sum_{v\in\cal{V}} \Pr[v\to \hat{v}_j \mid T] \cdot \vv{\pi}(v)}{\sum_{\hat{v}\in\hat{\cal{V}}}{\sum_{v\in\cal{V}} \Pr[v\to \hat{v} \mid T] \cdot \vv{\pi}(v)}}
\end{equation}

such that, $\sum_{j\in[n]}{\rho_j} = 1$.

Later, we show how assuming rational players, TRep games---with the utilities as below---ensure the reputation scores correspond to the relative trustworthiness.

\subsubsection{Utilities as a Function of Contribution}
We would like to reward players based on their relative \emph{contribution} to each server. To this extent, we analyze {\em Contribution PageRank} on $T'$. 
Observe that servers have contributions to each other indirectly through users. To exclude their effect on each other (as we are only interested in the users' contributions to the servers), we define relative contribution of $v$ to $\hat{v}$ \emph{restricted} to $\cal{V}$ as,
\begin{equation*}
    \vv{\omega_{\hat{v}\mid \cal{V}}^{-1}}(v) = \dfrac{\vv{\pi^{-1}_t}(v)}{\sum_{u\in \cal{V}} \vv{\pi^{-1}_{\hat{v}}}(u)}
\end{equation*}
 % we use the relative contribution with respect to only $\cal{V}$, $\vv{\omega_{\hat{v}\mid \cal{V}}^{-1}}(v)$.

% Therefore, the relative contribution, $\vv{\omega_{v}^{-1}}(u)$ is proportional to, $\frac{\Pr[u\to v\mid T]}{\vv{\pi}(v)}$
 
% Given that the reputation score is defined {\em linearly} with respect to each user, we can examine the ``contribution'' of  user $u$ to the final reputation (Personalized PageRank) of server $v$ as simply the ratio, $\dfrac{\Pr[u\to v]}{R(v)}$. This is precisely looking at the Contribution PageRank vector but $L_1$-normalized so that users' contributions can be measured relatively to each other.

We will use this quantity when determining the utilities of each player. Recall that in our model, the servers perform some prescribed task such that the behavior/correctness of the server can be evaluated. Specifically, we restrict ourselves to Bernoulli random variables with probabilities from $\vv{\sf{R}}$. For each server that behaves ``correctly'', we distribute a fraction of unit {\em reward} proportional to the ``contribution'' of the user to the server's reputation score.

And so, we derive the (expected) payoff of user $v$ in the graph as,
\begin{equation*}
\sum_{\hat{v}\in\hat{\cal{V}}}{\sf{R}_j \cdot {\vv{\omega_{\hat{v}\mid \cal{V}}^{-1}}(v)}}
\end{equation*}

When translating this to TRep games, $v_i \in \cal{V}$ translates to agent $P_i$, the transition probabilities map to (mixed) strategies, and the trustworthiness of the servers, $\hat{\cal{V}}$ map to nature's private state which fixes nature's strategy as $\vv{\sf{R}}$. Then, the utility is a function of the TRep graph induced by the strategies and nature's private state.
While players have beliefs about the trustworthiness of the servers (i.e., nature's state in the game), their utility does not depend on it. In full generality, the types can be used to establish a distribution over nature's state enabling Bayesian probability to be used in the analysis of the best response. Thus, we omit the types from the domain of the utility function for brevity.

Let $\vv{s} = \inp{\vv{s^{(1)}}, \ldots, \vv{s^{(n)}}} \in \Delta$ be a strategy profile for the players, and $\vv{r} = \inp{r_1, \ldots, r_m}\in\cal{N}$ be nature's strategy (which will be fixed). Then, we define the (expected) utility over mixed strategies for $P_i$, $u_i\colon \Delta \times \Delta_{\cal{N}}\to \inb{0,1}$,
\begin{equation*}
    \E\inb{u_i(\vv{s}, \vv{r})} = \sum_{j\in[m]}{{r_j \cdot {\vv{\omega_{\hat{v}_j\mid \cal{V}}^{-1}}(v_i)}}}
\end{equation*}
where the relative contribution vector is calculated over the induced TRep graph. 

% Recall that for the algorithm to converge or for the eigenvector of eigenvalue 1 to exist, the graph must be \rohit{list properties}.
% A directed edge-weighted bi-partite graph with edges going from one to another vertex set does not readily satisfy this property.

\subsection{TRep Games Under Different Classes of Beliefs}

We study the general TRep games by specifying different distributions of types (classes of beliefs) held by the players. For each distribution, we present an appropriate $\cal{E}$ and $f$ for which we can prove decodability. We reason that these choices are natural and highlight a powerful use case in PoR blockchains. 
Looking ahead, we will use the same decoding function for both games. Precisely, the decoding function will be the reputation score of the servers we derived using PageRank. 

We study the following specializations of TRep games,
\begin{enumerate}
    \item There is only one type: $\inp{\sf{R}_1, \ldots, \sf{R}_m}$ and is identically assigned to every player. This represents the case where every player has perfect information about nature. (ee Section~\ref{sec:perf-nash}.

    \item There exists a subset of players that are perfectly informed and thus, have type $\inp{\sf{R}_1, \ldots, \sf{R}_m}$. The remaining players are arbitrarily (mis-)informed---possess any type in $\cal{T}$. Only the perfectly informed players may play an action in $[m]$. See Section~\ref{sec:hierarchical-info}.

    \item There is an infinite number of types of the form, $\inp{\sf{R}_1 \pm\epsilon, \ldots, \sf{R}_m\pm\epsilon}$ for some small $\epsilon > 0$ but every player still has the same type. This is analogous to modeling the players beliefs as an \emph{additive noisy signal} over nature's state (See \Cref{sec:imperf-id-eps}).
    
    % \item For each $i\in[n]$, type $t_i \in T_i$ has the form, $\inp{\sf{R}_1 \pm\epsilon^{(i)}_1, \ldots, \sf{R}_m\pm\epsilon^{(i)}_m}$, and there exists $\epsilon > 0$ such that, for all $i\in[n]$, $j\in[m]$, $\epsilon^{(i)}_j \leq \epsilon$. Each player may not have the same type but the noise in each signal is upper bounded by some constant. 
    % (Section~\ref{sec:imperf-ub-eps})
\end{enumerate}

\section{Perfect Information About Nature}
\label{sec:perf-nash}

In this section, we study the first of the distributions mentioned above---every player has the same type that accurately capture nature's private state. Equivalently, every player has perfect information about nature's private state, the probability of its coin-flips, $\sf{R}_j$. In this case, we define the type space for every $P_i$ as $T_i = T_{\rm perfect} = \inb{0,1}^m$.

We denote this TRep game as $\cal{G}_{\rm perfect}$. 
Observe that $\cal{G}_{\rm perfect}$ can  be thought of as a (non-bayesian) game against nature with complete and perfect information, where every player aims to maximize, $\E\inb{u_i(\cdot, \vv{\sf{R}})}$.  (The expectation is over nature's randomness and the players' strategies.)

%, and the common prior as, $p(t_1, \ldots, t_n) = 1$ iff $t_1=\ldots= t_n$, and 0 otherwise. Therefore, $p(t^{(-i)}\mid t_i) = 1$ iff $t^{(-i)} = \inp{t_i, \ldots, t_i}$ (and 0 otherwise)

We make a few observations,
\begin{observation}
    $\cal{G}_{\rm perfect}$ is (totally) symmetric, and thus a symmetric Nash equilibrium (NE) exists \cite{Nash1951}.
\end{observation}

\begin{observation}
\label{obs:dominating}
For any agent, any strategy with non-zero probability on action $j\in \inset{m+1, \ldots, m+n}$ (i.e., picking another \emph{player} to endorse) is dominated by a strategy with zero probability on $j$. This is because the utility only takes into account the relative contribution to the servers and so the agent can increase their utility by reallocating the probability to some action in $\inset{1, \ldots, m}$.
\end{observation}

Thus, we may restrict ourselves to examining the graph where edges only go from $\cal{V}$ to $\hat{\cal{V}}$ (essentially making the graph bipartite).  In this setting, $\vv{\pi_{\hat{v}_j}^{-1}}(v_i)$ is directly proportional to $\Pr[v_i\to \hat{v}_j\mid T]$, and thus, 
\begin{equation*}
{\vv{\omega_{\hat{v}_j\mid \cal{V}}^{-1}}(v_i)} = \dfrac{\Pr[v_i\to \hat{v_j}\mid T]}{\sum_{v\in\cal{V}}{\Pr[v\to \hat{v_j}\mid T]}}
\end{equation*}

Let $\vv{s} = \inp{\vv{s^{(1)}}, \ldots, \vv{s^{(n)}}} \in \Delta$ be a strategy profile for the players, and $\vv{r} = \inp{r_1, \ldots, r_m}\in\cal{N}$ be nature's strategy. Then, we derive the (expected) utility over mixed strategies for $P_i$ using the correspondence in Section~\ref{sec:util-pr},
\begin{equation*}
    \E\inb{u_i(\vv{s}, \vv{r})} = \sum_{j\in[m]}{{r_i \cdot \dfrac{s^{(i)}_j}{\sum_{k\in[n]}{s^{(k)}_j}} }}
\end{equation*}

Denote $s^{(-i)}_j := \sum_{k\in[n]; k \neq i} {s^{(k)}_j}$, the sum of every other players' probability for action $j$.
Hence,
\begin{equation}
    \E\inb{u_i(\vv{s}, \vv{r})} = \sum_{j\in[m]}{r_i\cdot \dfrac{s^{(i)}_j}{s^{(i)}_j + s^{(-i)}_j}} \label{eq:exp_util}
\end{equation}

\begin{observation}
    $\cal{G}_{\rm perfect}$ is %{\em nearly} 
    (expected) constant-sum, with  $\sum_{j\in[m]}{\sf{R}_j}$ when players act rationally. What we mean here is that the total utility available to all players (the pot) is fixed in expectation. While there exists strategies where some of amount of the pot is ``wasted''---e.g. if $\sf{R}_1 > 0$ but no player picks 1 as their action--- this is  not rational assuming mixed strategies as a player ($P_i$) can  assign a tiny probability to action $a_i = 1$  and strictly improve their expected utility (even if it's split with other players). In other words, if any player plays a strategy such that the probability for each action $j$ (where $\sf{R}_j\neq 0$) is positive, then the game is constant-sum.
    % \rohit{$\cal{G}$ as is can also be made constant-sum by modifying the utility to add $H_j/n$ for every $j$ such that $m_j = 0$, i.e., we split the payoff uniformly for every $j$ that was not picked by any player. We forego this modification as we will later see that any strategy in equilibria... maybe remove this line altogether} 
\end{observation}

We show that $\cal{G}_{\rm perfect}$ is $\inp{\cal{E}_{\rm NE}, f_1}$-decodable for,
\begin{itemize}
    \item $\cal{E}_{\rm NE}$ equal to the {\em set of all (expected) Nash equilibria} of the game, and
    \item $f_1$ equals $\sf{N}$, { the L1-normalizing function}.
\end{itemize}

Observe that $f_1$ preserves ratios. That is, $f_1(\vv{\sf{R}}) = \sf{N}(\vv{\sf{R}})$, and 
%\begin{equation*}
 $   \dfrac{\sf{N}(\vv{\sf{R}})_i}{\sf{N}(\vv{\sf{R}})_j} = \dfrac{{\sf{R}}_i}{\sf{R}_j}$
%\end{equation*}

In words, any (expected) NE of the TRep game,
$\cal{G}_{\rm perfect}$ encodes the pairwise ratios of the components of nature's private state. %whchis makes $f_1$ order preserving. %This will be sufficient for extracting an ordering of the

We first prove the following lemma about the space of equilibria of $\cal{G}_{\rm{perfect}}$,
\begin{lemma}\label{lem:unique-nash}
    $\cal{G}_{\rm perfect}$ has a unique (expected) NE. The equilibrium is  the symmetric strategy profile, $\vv{s^*} = (\vv{{s^*}^{(1)}}, \ldots, \vv{{s^*}^{(n)}})$, where $\vv{{s^*}^{(i)}} = \sf{N}(\vv{\sf{R}})$ for all $i$. Therefore, $\cal{E}_{\rm NE} = \inset{\vv{s^{*}}}$.
\end{lemma}
\begin{proof}
Let $L = \inp{\sum_{j\in[m]}\sf{R}_j}/n$. Since the game is symmetric, $L$ is an upper bound on the (expected) payoff achievable by any player in a Nash equilibria. First, we show that every player can {\em guarantee} a minimum (expected) payoff of $L$.

Fix $P_i$ and its strategy, $\vv{s^{(i)}} = \sf{N}(\vv{\sf{R}})$. 

As a function of the strategies of the other players, the expected utility,
    \begin{alignat*}{2}
    \E\inb{u_i\inp{\sf{N}(\sf{R}), \vv{s^{(-i)}}; \vv{\sf{R}}}}
    &= 
    \sum_{j\in [m]} {\dfrac{\sf{N}(\sf{R})_j \cdot \sf{R}_j}{\vv{s^{(-i)}_j} + \sf{N}(\sf{R})_j}} \qquad &&\text{(using ~\eqref{eq:exp_util})} 
\end{alignat*}

%It is not hard to see that 
We observe that the Hermitian of the above function is negative definite and therefore is strictly convex as a function of $\vv{s^{(-i)}_j}$'s over the domain, \\$\inbr{\inp{\vv{s^{(-i)}_1, \ldots, \vv{s^{(-i)}_m}}} \ \colon\  \forall j\in[m], \vv{s^{(-i)}_j} \geq 0 \text{, and } \sum_{j\in[m]}{\vv{s^{(-i)}_j}} = (n-1)}$, which is bijective to $\Delta^{(-i)} = \prod_{\substack{j\in[m],j\neq i}}{\Delta_j}$, the strategy space of all other players.

Therefore, $\E\inb{u_i}$ has a {\em unique} global minimum with respect to the $\vv{s^{(-i)}_j}$'s. Partially differentiating in each variable, we find that the minimum is achieved at,
$\vv{s^{(-i)}_j} = (n-1)\sf{N}\inp{\vv{\sf{R}}}_j$. And therefore the minimum value of $\E\inb{u_i\inp{\vv{\sf{R}}; \sf{N}(\sf{R}), \vv{s^{(-i)}}}}$ is,
\begin{equation*}
    \sum_{j\in [m]} {\dfrac{\sf{N}(\sf{R})_j \cdot \sf{R}_j}{\vv{s^{(-i)}_j} + \sf{N}(\sf{R})_j}}
    = \sum_{j\in [m]} {\dfrac{\sf{N}(\sf{R})_j \cdot \sf{R}_j}{(n-1)\sf{N}\inp{\vv{\sf{R}}}_j + \sf{N}(\sf{R})_j}} = \sum_{j\in[m]}{\frac{\sf{R}_j}{n}} = L
\end{equation*}

Therefore, an (expected) payoff of $L$ is achievable by every player by playing $\sf{N}(\vv{\sf{R}})$ as its strategy. (It is {\em non-exploitable}.) Moreover if the other players play any strategy such that there exists $j$ and $s^{(-i)}_j \neq (n-1)\sf{N}(\vv{\sf{R}})_j$, then playing $\sf{N}(\vv{\sf{R}})$ guarantees payoff strictly greater than $L$. 
As any Nash cannot exceed individual payoff of more than $L$, we must have that if $\vv{c} = \inp{\vv{c^{(1)}}, \ldots, \vv{c^{(n)}}}$ is Nash then, $\vv{c^{(-i)}_j} = (n-1)\sf{N}\inp{\vv{\sf{R}}}_j$ for all $i\in[n], j\in[m]$.

This determines a full-rank linear system whose solution is %can be easily calculated,
$c^{(i)}_j = \sf{N}(\vv{\sf{R}})_j$, which is exactly $\vv{s^{*}}$. This implies, that if a Nash exists, it must be $\vv{s^{*}}$.%\qed

We can conclude that $\vv{s^*}$ must be Nash since we know that there exists a (symmetric) NE in symmetric games as per the observation above. This concludes the proof of the lemma. 
\end{proof}

 We show that the equilibria in TRep games can be decoded using PageRank to the relative values of nature's state (equivalently, to the relative trustworthiness of the servers).
 Recall that strategies in TRep games map to endorsements in TRep graphs. For each server $j$, we compute a reputation score, $\vrsc_j$ as per Eq.~\eqref{eq:rep_score} where the PageRank vector $\vv{\pi}$ is computed over the induced TRep graph and the mixed strategies correspond to transition probabilities.

Translating back to TRep games, we define the {\em PageRank (PR) decoding} function, $\cal{D}_{\rm PR}\colon \Delta\to \inb{0,1}^m$,

%\begin{alignat*}{1}
 \[   \cal{D}_{\rm PR}(\vv{s}) = \inp{\vrsc_1, \ldots, \vrsc_m}
    % \dfrac{1}{n}\cdot \inp{\sum_{i\in[n]}{\vv{s}^{(i)}_1}, \ldots, \sum_{i\in[n]}{\vv{s}^{(i)}_m}}
%\end{alignat*}
\textrm{, where } \vrsc_j = \dfrac{\sum_{i\in[n]} \vv{s}^{(i)}_j \cdot \vv{\pi}(v_i)}{\sum_{k\in[m]}\sum_{i\in[n]} \vv{s}^{(i)}_k \cdot \vv{\pi}(v_i)}
\]

% i.e., $D_{\rm PR}$ takes an average of the (mixed) strategies and is the TRep-game analogue of computing the reputation score.

Finally, we show,

\begin{theorem}\label{thm:perfect-decoder}
    $\cal{G}_{\rm perfect}$ is $\inp{\cal{E}_{\rm NE}, f_1}$-decodable using $D_{\rm PR}$.
\end{theorem}
\begin{proof}
Using Lemma~\ref{lem:unique-nash}, we have $\cal{E}_{\rm NE} = \inset{\inp{\sf{N}(\vv{\sf{R}}), \ldots, \sf{N}(\vv{\sf{R}})}}$, which implies,
\begin{alignat*}{1}
    \cal{D}_{\rm PR}(\cal{E}_{\rm NE})
    &= \inset{\inp{\dfrac{\sum_{i\in[n]} \sf{N}(\vv{\sf{R}})_1 \cdot \vv{\pi}(v_i)}{\sum_{k\in[m]}\sum_{i\in[n]} \sf{N}(\vv{\sf{R}})_k \cdot \vv{\pi}(v_i)}, \ldots, \dfrac{\sum_{i\in[n]} \sf{N}(\vv{\sf{R}})_m \cdot \vv{\pi}(v_i)}{\sum_{k\in[m]}\sum_{i\in[n]}\sf{N}(\vv{\sf{R}})_k \cdot \vv{\pi}(v_i)}}}
    \\
    &= \inset{\dfrac{\sum_{i\in[n]} \vv{\pi}(v_i)}{\sum_{k\in[m]}\sum_{i\in[n]} \sf{N}(\vv{\sf{R}})_k \cdot \vv{\pi}(v_i)}\cdot \inp{{\sf{N}(\vv{\sf{R}})}_1, \ldots, {\sf{N}(\vv{\sf{R}})}_m}}
%    \\
%    &= \inset{\dfrac{\sum_{i\in[n]} \vv{\pi}(v_i)}{\sum_{i\in[n]}\vv{\pi}(v_i) \sum_{k\in[m]} \sf{N}(\vv{\sf{R}})_k}\cdot \inp{{\sf{N}(\vv{\sf{R}})}_1, \ldots, {\sf{N}(\vv{\sf{R}})}_m}}
%    \\
%    &= \inset{\dfrac{1}{\sum_{k\in[m]} \sf{N}(\vv{\sf{R}})_k}\cdot \inp{{\sf{N}(\vv{\sf{R}})}_1, \ldots, {\sf{N}(\vv{\sf{R}})}_m}}
    \\
    &= \inset{\sf{N}(\vv{\sf{R}})}
%    \\
   % &
    = f_1(\vv{\sf{R}})
\end{alignat*}
\end{proof}

We remark that $\cal{G}_{\rm perfect}$ has the neat property that every players' strategy is exactly the function on nature's state we would like to decode.
Therefore, the decoding function could have also been to simply use any players' strategy in a Nash equilibrium as the output. While this is indeed true, this is merely a coincidence due to the simplified nature of the above game, which makes the symmetric NE being unique. The PR decoder provides a unifying treatment in which we use PageRank for both the utilities and decoding, and we believe this methodology can extend to more complicated graph structures.

% use a slightly different decoding function, the average of everyone's strategies. We will see that the same decoder can be used for the more complicated games in the next sections and thus, provides a more uniform approach inspired by PageRank to all TRep-games. \rohit{if we don't have game 3, can we still justify the averaging decoder??}

% We recall that the goal of our game is for rational users to (1) act in accordance to their \emph{beliefs}, $\gamma_i$, and (2) efficiently recover the relative trustworthiness of the reputation parties, $\hat{\cal{P}}$ as the reputation scores.

% In this setting, we show that there exists a \emph{unique} Nash equilibrium (Lemma~\ref{}) and an efficient and natural decoder (Lemma~\ref{}) that recovers the relative trustworthiness of $\hat{\cal{P}}$ from $\cal{P}$'s endorsements. \rohit{define relative ordering in prelims/overview. no longer need definition of inversions i guess. or no inversions is a corollary. actually, we prob need it for the lottery later...}
\section{Hierarchical Information: Leveraging Perfect Information}
\label{sec:hierarchical-info}

We examine TRep games in which the graph structure does not reduce to bi-partite when there is an asymmetry of information about nature's state.
In particular, there exists a set of $k$ agents, $\cal{P}_\mathrm{perfect} \subseteq \cal{P}$ that have perfect information on (all of) nature's state.
While the other agents may have some partial or no information on nature.
We believe this to be an interesting specialization of TRep games to capture newly joining, fresh parties in the blockchain network that may not have enough information about the servers but may have beliefs about existing (perfect information) players. 
We apply the restriction that the other agents, $\cal{P}_\mathrm{fresh} = \cal{P}\setminus \cal{P}_\mathrm{perfect}$ may not directly endorse a server and can only endorse other agents.
Without loss of generality, let $\cal{P}_\mathrm{perfect} = \inset{P_1, \ldots, P_k}$ and $\cal{P}_\mathrm{fresh} = \inset{P_{k+1}, \ldots, P_n}$. Then, we have $\cal{A}_i = [m+n]$ for $i\in[k]$ and $\cal{A}_i = \inset{m+1, \ldots, n}$ (the set of actions that endorse players) for $i \in\inset{k+1, \ldots, n}$.
Denote this game as $\cal{G}_\mathrm{hierarchy}$.

We show that $\cal{G}_\mathrm{hierarchy}$ is also $\inp{\cal{E}_{\rm NE}, f_1}$-decodable! Moreover, it turns out that in the NE strategies, $\cal{P}_\mathrm{perfect}$ follow the same strategy as in the perfect information case \emph{invariant} of $\cal{P}_\mathrm{fresh}$'s strategies. 
Similar to Observation~\ref{obs:dominating}, we observe,

\begin{observation}
For any agent, $P_i\in \cal{P}_\mathrm{perfect}$, any strategy with non-zero probability on action $j\in\inset{m+1, \ldots, m+n}$ is dominated by a strategy with zero probability on~$j$.
\end{observation}

Therefore, in any Nash equilibrium, $\cal{P}_\mathrm{perfect}$ only point to the servers (strategies over $[m]$). Let $\cal{V}_\mathrm{perfect} = \inset{v_1, \ldots, v_k}$ and $\cal{V}_\mathrm{fresh} = \inset{v_{k+1}, \ldots, v_n}$ respectively correspond to the vertices to the perfect information and other agents in the induced TRep graph. The contribution of $P_i\in \cal{P}_\mathrm{perfect}$ in the induced TRep graph (and its markov chain, $T$),

\begin{equation*}
\vv{\pi_{\hat{v}_j}^{-1}}(v_i) = (1-\alpha)\Pr\inb{v_i\to \hat{v}_j \mid T}
\end{equation*}

We obtain the relative contribution of the perfect information agents,
\begin{equation}
\vv{\omega_{\hat{v}_j\mid \cal{V}}^{-1}}(v_i)
= \dfrac{\vv{\pi^{-1}_{\hat{v}_j}}(v_i)}{\sum_{v\in \cal{V}} \vv{\pi^{-1}_{\hat{v}_j}}(v)}
= \dfrac{(1-\alpha)\Pr\inb{v_i\to \hat{v}_j \mid T}}{\sum_{v\in\cal{V}_\mathrm{perfect}}{(1-\alpha)\Pr\inb{v\to \hat{v}_j \mid T}} + 
\sum_{v\in\cal{V}_\mathrm{fresh}} \vv{\pi^{-1}_{\hat{v}_j}}(v)}
\end{equation}

We can use a similar proof as for Lemma~\ref{lem:unique-nash} to show that the strategy of perfect information agents, $\vv{s^{(i)}} = \sf{N}(\vv{\sf{R}})$ in \emph{any} Nash equilibrium.

This is sufficient to prove $f_1$-decodability of $\cal{G}_\mathrm{hierarchy}$. For any server $\hat{v}_j\in\hat{\cal{P}}$, recall its reputation score (Eq.~\eqref{eq:rep_score}),
\begin{equation*}
    \vrsc_j = \dfrac{ \sum_{v\in\cal{V}} \Pr[v\to \hat{v}_j \mid T] \cdot \vv{\pi}(v)}{\sum_{\hat{v}\in\hat{\cal{V}}}{\sum_{v\in\cal{V}} \Pr[v\to \hat{v} \mid T] \cdot \vv{\pi}(v)}} = \dfrac{ \sum_{v\in\cal{V}_\mathrm{perfect}} \Pr[v\to \hat{v}_j \mid T] \cdot \vv{\pi}(v)}{\sum_{\hat{v}\in\hat{\cal{V}}}{\sum_{v\in\cal{V}_\mathrm{perfect}} \Pr[v\to \hat{v} \mid T] \cdot \vv{\pi}(v)}}
\end{equation*}
where the sum is simplified to only be over $\cal{P}_\mathrm{perfect}$ since we disallow edges from $\cal{P}_\mathrm{fresh}$ to $\hat{\cal{P}}$. Thus, the same PageRank decoder, $\cal{D}_{\mathrm{PR}}$ works just as in Theorem~\ref{thm:perfect-decoder} regardless of $\cal{P}_\mathrm{fresh}$'s strategies.

We conclude with a remark about the equilibrium strategies of $\cal{P}_\mathrm{fresh}$. While  decodability does not depend on $\cal{P}_\mathrm{fresh}$, not every strategy maximizes their utility. In fact, analogous to Observation~\ref{obs:dominating}, %we have that 
any strategy with non-zero probability on action $j\in\inset{m+k+1, \ldots, m+n}$ (i.e., endorsing another ``fresh'' player) is dominated by a strategy with zero probability on $j$. Thus, in any NE, players in $\cal{P}_\mathrm{fresh}$ only endorse $\cal{P}_\mathrm{perfect}$ and \emph{not} each other. One can show that \emph{any} strategy subject to this restriction is Nash. 
 \section{(Consistent) Noisy Information About Nature}
\label{sec:imperf-id-eps}

We continue with our study of TRep games and explore decodability when every player has an approximation of nature's private state, $\sf{R'} = \inp{\sf{R'}_1,\ldots, \sf{R'}_m}$. We suppose that $\sf{R'}_j\in[0,1]$ is distributed such that its expecation as a random variable is equal to $\sf{R}_j$, and there exists a constant, $\epsilon > 0$, and some probability, $p$, such that $\Pr\inb{\sf{R'}_j\notin\inb{\sf{R_j}-\epsilon, \sf{R_j}+\epsilon}}\leq p$. 
 We model it as such to capture the idea of confidence intervals when approximating a random variable. E.g., the sample mean is centered around the true value and is normally distributed.
Write, $\sf{R'}_j =  \sf{R}_j+\delta_j\epsilon'$ for $\delta_j\in\inset{-1,1}$.

We define the type space for every player, $P_i$, as $T_i = T_{\rm approx} = \inb{0,1}^m$.
We assume that every player has the {\em same} type, $\vv{t_i} = \vv{t_j} = \vv{\sf{R'}} = \inp{\sf{R'}_1,\ldots, \sf{R'}_m}$, i.e. every player has the same belief about nature's state.
Denote this game as $\cal{G}_{\rm noisy}$. The same observations above hold for $\cal{G}_{\rm noisy}$ as well.

% The players may know $\epsilon$ and $p$ is---e.g. if nature's state was estimated using a simple sample mean or some other statistic cite [approximation  algorithms].
As we shall see, $\epsilon$ and $p$
%our analysis of the game is {\em independent} of $\epsilon$ and $p$ \rohit{\textbf{IS THIS TRUE???}}; however, they$
will determine the ``slack'' in the relative trustworthiness that we're able to decode.

% More formally, TR games with uncertain beliefs about nature, $\cal{G}_{\rm noisy-id}$, can be modeled as a {\em ex-ante Bayesian game} where each player has the same {\em type} but no player knows their own type or their opponents'.
% In this case the type can be specified by a value, $\epsilon \geq 0$ and an $m$-vector in $\inset{+1,-1}^m$ specifying whether the noise in nature's private state is $\pm\epsilon$ for each coin flip respectively. For some fixed $\epsilon$, nature randomly selects a vector in $\inset{+1,-1}^m$ and assigns this type to every player.
% \begin{definition}[$\cal{G}_{\rm noisy-id}$]
%     We define $\cal{G}_{\rm noisy-id}$ as a TR game with the following modification,
%     \begin{itemize}
%         \item \textbf{Types.} Let $T = \inb{0,1} \times \inset{+1,-1}^m$ be the space of all types. Each player, $P_i$ has the same type space, $T_i = T$. 
%         E.g., the type $\inp{\epsilon, \vv{k} = \inp{k_1,\ldots, k_m}}\in T$ is interpreted as a vector of noise applied to the state of nature resulting in, $\inp{\sf{R}_1+k_1\epsilon, \ldots, \sf{R}_m+k_m\epsilon}$.

%         \item \textbf{Utilities.} We redefine the utilities to include the type of the player. Let $u_i\colon T_i \times \Delta_i \to $
%     \end{itemize}

%     ... Therefore, $\cal{G}$ is an ex-ante Bayesian game against nature.
% \end{definition}

In Section~\ref{sec:perf-nash}, we showed that in the perfect information setting ($\cal{G}_{\rm perfect}$), where every player a priori precisely knew nature's private state, the unique NE encoded the relative values of nature's private state, $\vv{\sf{R}}$. This is of course still true in $\cal{G}_{\rm noisy}$ since the utilities have not changed; however, the players operate with incomplete information and so, we must deal with a weaker notion than Nash.

Instead, we show that TRep games with the noisy but identical type distribution shown above are $\inp{\cal{E}_{\rm tt}, f_2}$-decodable for,
\begin{itemize}
    \item $\cal{E}_{\rm tt}$ is the singular set of the ``truth-telling'' strategy profile, denoted by $\vv{s_{\rm tt}}$, in which every players play $\sf{N}(\vv{t_i})$ as their strategy. We call the strategy $\sf{N}(\vv{t_i})$ as ``truth-telling'' as it represents some fixed function of the players' belief about the state of nature, which it can locally compute. 
    % \item $f_2\colon\inb{0,1}^m \to \inb{0,1}^m$ is a randomized function that takes as input $\vv{u} = \inp{u_1,\ldots,u_m}$, randomly selects a vector, $\vv{k} = \inp{k_1,\ldots,k_m}\in\inset{+1,-1}^m$, and outputs $\vv{v}=\inp{v_1,\ldots,v_m}$ such that $\dfrac{v_i}{v_j} = \dfrac{u_i+k_i\epsilon}{u_j+k_j\epsilon}$.
 %    \item $f_2\colon \inb{0,1}^m \to \inb{0,1}^m$ is a randomized function that takes as input $\vv{u} = \inp{u_1,\ldots,u_m}$, and outputs $\vv{v}$ such that  $\E\inb{\norm{\infty}{\vv{v}-\sf{N}(\vv{u})}} \leq \epsilon'$, where $\epsilon' = \epsilon'(\epsilon)$ is some scalar that depends on $\epsilon$. In words, $f_3$, in expectation, outputs a vector ``close'' to the {\em $L_1$-normalized} vector with respect to the $L_{\infty}$-norm.
    \item $f_2\colon \inb{0,1}^m \to \inb{0,1}^m$ is a randomized function that takes as input $\vv{u} = \inp{u_1,\ldots,u_m}$, and outputs $\vv{v}$ such that  $\Pr\inb{\E\inb{\norm{\infty}{\vv{v}-\sf{N}(\vv{u})}} \leq \frac{\epsilon}{\norm{1}{\vv{\sf{R}}}}}\geq 1-mp - q$, where $q$ is an
    upper bound on the probability that $\abs{\norm{1}{\vv{\sf{R'}}} - \norm{1}{\vv{\sf{R}}}} \geq \delta$ for some small constant, $\delta$. We show that $q$ is exponentially-decreasing in $m$.
     In words, with high probability, $f_2$, in expectation, outputs a vector ``close'' to the {\em $L_1$-normalized} vector with respect to the $L_{\infty}$-norm.
\end{itemize}

% The proof can be found in Appendix~\ref{sec:proofs}.

In the following we show that $\vv{s_{\rm tt}}$,  is an $\epsilon'$-(Nash) equilibrium \cite{10.5555/1296179} for $\cal{G}_{\rm noisy}$, where $\epsilon'$ is decreasing in $n/m^2$. 

\begin{lemma}\label{lem:eps-nash}
    Assuming $\epsilon=O(1/n)$, the ``truth-telling'' strategy profile, $\vv{s_{\rm tt}}$, is $\epsilon'$-NE for $\epsilon' =  O(m^2/n)$.
\end{lemma}
\begin{proof}
%We show that $\vv{s_{\rm tt}}$ is an $\epsilon'$-(Nash) equilibrium \cite{10.5555/1296179} for $\cal{G}_{\rm noisy}$, where $\epsilon'$ is decreasing in $n/m^2$. That is, 
All players playing $\sf{N}(\vv{\sf{R}})$ guarantees a payoff %that is 
at most $\epsilon'$ less than the  %from the 
best-response payoff. $P_i$'s (expected) utility as a function of everyone else playing truthfully is,
\begin{alignat*}{2}
\E\inb{u_i(\vv{x}, {s_{\rm tt}}_{-i}; \vv{\sf{R}})}
&= 
\sum_{j\in [m]} {\dfrac{x_j \cdot \sf{R}_j}{(n-1)\sf{N}(\vv{\sf{R'}})_j + x_j}}
\end{alignat*}

The above function is strictly concave on the domain of $\vv{x}$: $\inbr{x_i \geq 0; \sum_{i\in[n]}{x_i} = 1}$, and therefore obtains its maximum at a unique point. Partially differentiating and simplifying yields maximum value at $\vv{x^*}$,
\begin{equation*}
x^*_j = n\dfrac{\sqrt{\sf{R}_j\sf{R'}_j}}{\sum_{k\in[m]}{\sqrt{\sf{R}_k\sf{R'}_k}}} - (n-1)\sf{N}(\vv{\sf{R'}})_j
\end{equation*}

and thus, the deviation from the truthful strategy in each component, 
\begin{equation*}
\Delta x_j = x^*_j - \sf{N}(\vv{\sf{R'}})_j = n\dfrac{\sqrt{\sf{R}_j\sf{R'}_j}}{\sum_{k\in[m]}{\sqrt{\sf{R}_k\sf{R'}_k}}} - n\dfrac{\sf{R'}_j}{\norm{1}{\vv{\sf{R'}}}}
\end{equation*}

We coarsely approximate $\epsilon'$. For  $\epsilon=O(1/n)$, $\Delta x_j$ is small, and we can approximate the change in utility, $\epsilon' = \E\inb{u_i(\vv{x^*}, {s_{\rm tt}}_{-i}; \vv{\sf{R}})} - \E\inb{u_i(\vv{s_{\rm tt}}; \vv{\sf{R}})}$ using linear approximation. We obtain, \\
$\epsilon' \leq \dfrac{m^2(n-1)}{n^2}\dfrac{1+\epsilon}{1-\epsilon}$.  This completes the proof of the lemma.
\end{proof}

We remark that the above is an asymptotic bound that aims at showing that for small enough $\epsilon$, $\epsilon'$ is inversely proportional to $n/m^2$. In particular, for our blockchain-motivated scenario, where a natural assumption is that $n\gg m$ (as is the case in  major cryptocurrencies) $\epsilon'$  will be monotonically decreasing as the system scales.  We note in passing that our numerical experiments demonstrate that the value of $\epsilon'$ would typically be much smaller than $m^2/n$, even for a large number of users. This is consistent with our intuition that in our game (where utilities are normalized) %$$
%
%
% and a large number of cl
%
%Therefore,  
playing one's beliefs guarantees a payoff close to the best possible payoff one can hope even when everyone is perfectly informed; therefore small deviations from this strategy should not yield major gains in utility. 
%
%given the nature of the game is the best payoff one can hope, for as you do not have perfect information about nature's state and thus, its (fixed random) strategy.
We conjecture that %believe it could be possible to prove that 
playing the truth-telling strategy is  an {\em approximate Bayesian Nash Equilibrium} for a tighter $\epsilon'$ (dominated by  $\epsilon$ and $p$). We leave this as a question for future work. %\rohit{vassilis, is this okay to say???} 

% Bayesian Nash Equilibria is an extremely natural class of strategy profiles for bayesian games. Intuitively, it captures the notion that players play the best response with respect to the other players beliefs and strategies.
% We remark that since each player has the same type which is known to everyone, $\cal{G}_{\rm noisy-id}$ reduces to $\cal{G}_{\rm perfect}$ with an alteration of nature in which its state is $\vv{\sf{R}'}$ instead of $\vv{\sf{R}}$. In other words, any Bayesian Nash equilibrium of $\cal{G}_{\rm noisy-id}$ is simply a Nash equilibrium in $\cal{G}_{\rm perfect}$ with a different (perceived) state of nature. Indeed, the (expected) utility for every $P_i$,
% \begin{align*}
% \E\inb{u_i(\vv{R'}; \vv{s})}
% &= \sum_{j\in[m]}{R'_j \cdot \dfrac{s^{(i)}_j}{s^{(i)}_j + s^{(-i)}_j}}
% \end{align*}
% Therefore, we present the following lemma whose proof is identical to that of Lemma~\ref{lem:unique-nash}.
% \begin{lemma}\label{lem:unique-b-nash}
%     $\cal{G}_{\rm noisy-id}$ has a unique (expected Bayesian) Nash equilibrium. The Nash is precisely the symmetric strategy profile, $\vv{s^*} = (s^*_1, \ldots, s^*_n)$, where $s^*_i = \sf{N}(\vv{\sf{R'}})$ for all $i$. Therefore, $\cal{E}_{\rm BNE} = \inset{\vv{s^{*}}}$.
% \end{lemma}

Now, unlike $\cal{G}_{\rm perfect}$, we cannot perfectly decode $f_1(\vv{\sf{R}}) = \sf{N}(\vv{\sf{R}})$ since the players only know an approximation of nature's state . Instead, we prove decodability for the ``weaker'' function, $f_2$, that only approximates the $L_1$-normalized function in the {$L_\infty$}-norm. Again, we will use the same decoding function, $D_{\rm PR}$ to show that PageRank approximately decodes to the relative trustworthiness scores of the servers.

   We first show that $\Pr\inb{\abs{\norm{1}{\vv{\sf{R'}}} - \norm{1}{\vv{\sf{R}}}} \geq \delta} \leq q$ for some $q$ that is exponentially small in $m$ using Hoeffding's bound.
    % Observe that $\Pr\inb{\abs{\norm{1}{\vv{\sf{R'}}} - \norm{1}{\vv{\sf{R}}}}} = \Pr\inb{\norm{1}{\vv{\sf{R'}}} \geq \norm{1}{\vv{\sf{R}}} + \delta} + \Pr\inb{\norm{1}{\vv{\sf{R'}}} \leq \norm{1}{\vv{\sf{R}}} - \delta}$.
    Formally, we treat $\sf{R'}_j$ as a random variable. By assumption, we have $\E\inb{\sf{R'}_j} = \sf{R_j}$. Also, $\norm{1}{\vv{\sf{R'}}} = \sum_{j\in[m]}{\sf{R'}_j}$ by definition, which implies 
    
\begin{equation}\label{veq}\E\inb{\norm{1}{\vv{\sf{R'}}}} = \E\inb{\sum{\sf{R'}_j}} = \sum{\E\inb{\sf{R'}_j}} = \sum{\sf{R}_j} = \norm{1}{\vv{\sf{R}}}.\end{equation}
    % Formally, we can treat $\sf{R'}_j \in \cal{U}\inb{\sf{R}_j-\epsilon, \sf{R}_j+\epsilon}$ as  (independent) uniformly random variables with probability at least $1-mp$ by union bound over $j$. Denote this event as $E$. Therefore, $\Pr[E] \geq (1-mp)$.
    Therefore, $\Pr\inb{\abs{\norm{1}{\vv{\sf{R'}}} - \norm{1}{\vv{\sf{R}}}} \geq \delta} = \Pr\inb{\abs{\norm{1}{\vv{\sf{R'}}} - \E\inb{\norm{1}{\vv{\sf{R'}}}}} \geq \delta}$, which by Hoeffding's bound \cite{Hoeffding1963},
\begin{equation*}
    \Pr\inb{\abs{\norm{1}{\vv{\sf{R'}}} - \E\inb{\norm{1}{\vv{\sf{R'}}}}} \geq \delta}
    \leq {\rm exp}\inp{-\frac{\delta^2}{4\epsilon^2 m}} = q
\end{equation*}
which is exponentially decreasing in $m$. Since the above holds for any $\delta$,  for sufficiently large $m$, we can approximate,
\begin{equation}\label{eq:approx}
\norm{1}{\vv{\sf{R'}}} \approx \E[\norm{1}{\vv{\sf{R'}}}] \overset{Eq.~\eqref{veq}}{=} \norm{1}{\vv{\sf{R}}}.
\end{equation}
%\end{remark}

Using the above, we show,
\begin{theorem}
    $\cal{G}_{\rm noisy}$ is $\inp{\cal{E}_{\rm tt}, f_2}$-decodable.
\end{theorem}
\begin{proof} 
    We have, $\cal{E}_{\rm tt} = \inset{\inp{\sf{N}(\vv{\sf{R'}}), \ldots, \sf{N}(\vv{\sf{R'}})}}$, which implies,
    $
    \cal{D}_{\rm PR}(\cal{E})
    = \inset{\sf{N}(\vv{\sf{R'}})}
$

We calculate,
\begin{alignat*}{2}
\E\inb{\norm{\infty}{\cal{D}_{\rm PR}\inp{\cal{E}_{\rm tt}} - \sf{N}(\vv{\sf{R}})}}
&= 
\E\inb{\norm{\infty}{\sf{N}(\vv{\sf{R'}}) - \sf{N}(\vv{\sf{R}})}}
= \max_{j\in[m]}{\E\inb{\abs{\sf{N}(\vv{\sf{R'}})_j - \sf{N}(\vv{\sf{R}})_j}}} 
\end{alignat*}

We restrict ourselves to the intersection of events, \\$E = \inset{\sf{R'}_j\in\inb{\sf{R_j}-\epsilon, \sf{R_j}+\epsilon} \text { for all } j\in[m]}$, and $H = \inset{\abs{\norm{1}{\vv{\sf{R'}}} - \E\inb{\norm{1}{\vv{\sf{R'}}}}} \leq \delta}$.

Now,
$
    \Pr[E\cap H] \geq 1 - \Pr[E^c \cup H^c]
    \geq 1 - (\Pr[E^c] + \Pr[H^c])
$, where $A^c$ is defined the complement of event $A$. By union bound, $\Pr[E^c] \leq mp$, and, $\Pr[H^c] \leq q$ as we showed above. Thus, $\Pr[E\cap H] \geq 1 - mp - q$.

Conditioned on $E\cap H$, we have $\epsilon' \leq \epsilon$ and,

\begin{alignat*}{2}
\max_{j\in[m]}{\E\inb{\abs{\sf{N}(\vv{\sf{R'}})_j - \sf{N}(\vv{\sf{R}})_j}}} 
&\leq \max_{j\in[m]}{\E\inb{\abs{\frac{\sf{R}_j+\delta_j\epsilon'}{\norm{1}{\vv{\sf{R'}}}} - \frac{\sf{R}_j}{\norm{1}{\vv{\sf{R}}}}}}} 
\\
&\hspace*{-1.5ex}\overset{Eq.~\eqref{eq:approx}}{\approx} \max_{j\in[m]}{\E\inb{\abs{\frac{\sf{R}_j+\delta_j\epsilon'}{\norm{1}{\vv{\sf{R}}}} - \frac{\sf{R}_j}{\norm{1}{\vv{\sf{R}}}}}}}
\\
&\leq \dfrac{\epsilon}{\E\inb{\norm{1}{\vv{\sf{R}}}}}
\end{alignat*}

We conclude that $\E\inb{\norm{\infty}{\cal{D}_{\rm PR}\inp{\cal{E}_{\rm tt}} - \sf{N}(\vv{\sf{R}})}} \leq \dfrac{\epsilon}{\E\inb{\norm{1}{\vv{\sf{R}}}}}$ with probability %at least 
$\geq 1-mp-q$.
\end{proof}

Intuitively, this demonstrates that for sufficiently small $p$, i.e., sufficiently high confidence on the interval, as the system scales in $m$, the decoding of the above approximate Nash becomes order preserving with respect to the ground truth. In fact, assuming sufficient large gaps in the ground truth, the above will be true even for small values of $m$. 

% \rohit{Vassilis, we should talk about what happens when there are a lot of less trustworthy vs more trustworthy servers. E.g., if there are fewer trustworthy servers, then the expectation is small and thus, we decode something more noisy. But that noise should be okay I think...}

% \rohit{maybe we need to study Bayesian Nash Equilibrium and show that these equilibria are a subset of the $\epsilon$-Nash equilibria... there are two possibilities: (1) ex-ante stage game, no player knows their type. here type exactly means the value of $\epsilon$ and in which direction $\pm$. this seems to align well with our model (2) interim stage: player knows their own type but only a distribution over the other players' types. this might make more sense for the upper bounded game case where based on when i joined the game, i know how good of an approximation i have so i know something about my type but i don't know it exactly and maybe I can assume everyone else's type is within the upper bound}
% \input{game3}
% \input{dynpart}
\section{Conclusion and Future Work}
\label{sec:conclusion}
We introduced a class of games, called trustworthy reputation games, which is motivated by the problem of a blockchain system publicly extracting its users' collective perception of the trustworthiness of the blockchain nodes. Our games make use of the PageRank algorithm for both ranking and rewarding nodes. Importantly, we demonstrate that this use comes with natural guarantees about how accurately our extraction works. In the process, we initiate the study of benchmarking PageRank in trustworthiness settings and analyze its behavior. Our work opens a number of interesting questions, including how to extend our treatment to repeated games, as a means of allowing dynamically updateable reputation systems, and investigating the behavior of our PageRank model on different classes of graphs and topologies both empirically and theoretically. We leave it as future work to explore the sensitivity of ranks to misinformed users.  

% \rohit{UPDATE conclusion to mention future work on imperfect information's effect on the rankings... maybe experiments, bounds on difference, et cetera. }

% Bibliography
\bibliographystyle{splncs04}
\bibliography{reputation}

\newpage
\begin{appendix}
\section{Connection to PoR/PoS Blockchain}\label{sec:blockchain}
We show how our TRep games can be applied to PoR blockchains, e.g.,~\cite{KleOstZik20}. First we give a high-level overview of that paradigm {\em proof-of-reputation with proof-of-stake fallback} (in short, PoR/PoS) paradigm (\ref{sec:PoR/PoS}). Subsequently, in \ref{sec:use_rep_in_bc}, we highlight how to use TRep games in that paradigm.
% the PoR with PoS backup construction. 

\subsection{Proof of Reputation with Proof of Stake Fallback}\label{sec:PoR/PoS}
%In Section~\ref{sec:blockchain}, we will show how our trustworthy reputation games can be applied to the  {\em proof-of-reputation with proof-of-stake fallback} (in short, PoR/PoS) paradigm for bootstrapping---i.e., setting the inital parameters of---the blockchain proposed in~\cite{KleOstZik20}. Below, we provide a high level overview of that paradigm. 

The PoR/PoS paradigm assumes a so-called  {\em reputation system}, which is a vector of $n$  independent (probability distributions on) binary random variables $(\vrep_1, \ldots, \vrep_n)$. The number $n$ is the number of blockchain nodes (also referred to as 
 {\em servers}) who are tasked  with running the blockchain protocol, e.g., in Bitcoin that would be the {\em miners.} Each $\vrep_i$ corresponds to the probability that the $i$th node will remain ``honest'' in the protocol, i.e., it will follow its prescribed protocol. 
Given such a reputation vector, as long as it closely captures the ``ground truth,'' or in other words, the true probability of servers' honesty, \cite{KleOstZik20} shows how to select committees of size $C= {\rm polylog}(n)$ so that with  big probability a majority of the parties in all committees will be honest. 

%We note in passing that selection of honest majority committees is essential for blockchains that following the 
%parties (i.e., parties that follow the protocol) as possible. 

As argued in~\cite{KleOstZik20,AC:AshLin13} the simplest way for achieving the above goal (maximizes the probability of honest majorities) is to order the parties according to their reputation, and select the top $C$ in this ordering. (In fact, the actual mechanism/lottery for this selection has several additional properties, which ensures and intuitive notion of fairness; but the above simple deterministic selection algorithm is sufficient for understanding the use of reputation within PoR blockchains.) 

Given such a mechanism for selecting honest-majority committees, a PoR blockchain can be constructed in a similar way as common Byzantine Fault-Tolerant (BFT) blockchains, e.g., Algorand in the proof-of-stake (PoS) setting: proceed in phases (often referred to as {\em slots} or {\em blockchain rounds}) where in each $i$th round a committee is chosen to vote on (by adding their digital signature on) the $i$th block, and a block is accepted if and only if its voted by more than $C/2$ parties in this $i$th-slot committee (i.e., has a majority vote by committee members). 

Importantly, observe that in order to take a decision on whether or not the majority of the $i$th slot committee has voted, it is essential that parties can verify whether a signature corresponds to a party in this committee. And to make sure that parties adopt the same block, there should be agreement among the committee members. This is done by making sure that the reputation system (and the associated randomness) used for the lottery are ``known'' to the blockchain (i.e., they are encoded in its past blocks). 

The blockchain from~\cite{KleOstZik20} also fortifies the security of its above PoR methodology by assuming a fallback blockchain, which is based on the proof-of-stake paradigm, and is used to detect and correct forks due to an inaccurate reputation system. In particular, parties running the PoR based constriction above, periodically report (a publicly verifiable digest) of their view. The proposed mechanisms ensures that if the blockchain properties are violated then it will be promptly noted on the secondary chain. In this case, the system (temporarily) falls back to that secondary chain.  As discussed above, this mechanism is necessary for security; however, its details beyond what is discussed above are not relevant for our paper.

% Assume a correlation-free static reputation system \rohit{change wording depending on how much notation we want to borrow from [KOZ]} wherein each reputation party, $\hat{P}_i$, has reputation, $\sf{R}_i$, i.e., $$\Pr\inb{\hat{P}_i \ \text{behaves \emph{honestly}}} = \sf{R}_i$$

% Here, \emph{honestly} means with respect to some oracle that say, outputs 1 when the party behaves correctly, and 0 otherwise. Therefore, $\sf{R}_i$ would be the probability that the oracle outputs 1 for $\hat{P}_i$.

%begin Cancel 1
\cancel{ Independently, we study the underlying problem of \emph{discovering} $\vv{\sf{R}}$ in the decentralized setting. Assuming rational players, we construct an incentive mechanism and system such that the output prescribed by the \emph{optimal} actions of each party (say, in a Nash equilibrium) is equal or closely approximates the desired function computed over $\vv{\sf{R}}$.

We assume that every party behaves rationally---decisions are made with respect to maximizing the expected reward/payoff.
Informally, we define an incentive mechanism to be a game-theoretic mechanism where every epoch, the action taken by each party is prescribed by the Nash equilibrium of the game.

More generally, we would hope for an incentive mechanism in which parties are incentivized to report their \emph{true beliefs}. In a setting in which each party has full information, this would be exactly proportional to $\sf{Rep}$.

Note that if there are multiple, then we implicitly pick the one with the maximum total sum of reward for the malicious parties \rohit{this is ill-defined}.
}%end Cancel 1

%Since the novelty of our work is not the application but rather TRep games themselves, the purpose of this section is to demonstrate the connection, rather than providing the concrete blockchain construction. For this reason,  we will keep the discussion informal. 

\subsection{Using Trustworthy Reputation Games in a PoR/PoS Blockchain}\label{sec:use_rep_in_bc}

In a (byzantine-fault tolerant) PoR blockchain like~\cite{KleOstZik20}, a committee of nodes is selected for proposing and voting on each block. %, and (in order to incentivize participation). 
Most blockchain-based cryptocurrency systems, give rewards (in terms of coins) to the members of such committee to incentivize participation. However, the rest of the users of the system typically, do not receive any rewards. An exception here is systems that offer their users ``dividends'' or interest for participation in order to boost adoption and/or availability (e.g., Algorand) and Proof-of-Stake blockchains that support stake delegation (e.g., Cardano, Ethereum, etc.)  where users can delegate their stake to stake-pools and they get a fraction of their rewards when their stake pool (operator) is selected to propose the next block. We conjecture that our proposed mechanisms can be used also to incentivize truthful stake delegation. We view this as an interesting research question, albeit less relevant than reputation in a PoR system, as untruthfully extracted reputation can hurt the systems security, whereas untruthful stake delegation does not: if the  majority of stake is in honest hands, then honest parties will delegate truthfully, otherwise, security cannot be ensured. %either is or is not in the hands of honest users, whether delegators or pool operators)

Here, we show how to encode our payoff function into the PoR blockchain for rewarding the users/endorsers. 
%
%
%
%in ways that allow everyone to sufficiently estimate their belief of the ground truth. %---and incentivization truthfulness/ %---which fit the goals of PoR and show their connection to our TRep Games. 
Consistently with~\cite{KleOstZik20}, we will assume the blockchain leverages a static trustworthy reputation system---i.e., the trustworthiness of the nodes is fixed (part of the ground truth) and recorded on the PoR blockchain's genesis block. 

We recall how~\cite{KleOstZik20} encodes such static reputation:  The honesty probability of each party is encoded in the ground truth (which we referred to here as trustworthiness score vector). In particular,
at the beginning of the system's execution, a biased coin is (privately) flipped for each node $s_i$ which is heads with probability $\vrep_i$. If the outcome of the coin is heads, then $s_i$ is corrupted/Byzantine. We will refer to parties who are not corrupted as {\em honest}. We remark that an honest party follows the protocol, whereas a byzantine party might not. To connect the cryptographic notion of corruption with the notion of honesty as used in trustworthy reputation we make the following natural convention: all byzantine-corrupted parties are {\em eventually faulty}, in the sense that there is a bound $\lambda\in\mathbb{N}$ such that all corrupted parties will be faulty within $\lambda$ rounds of the protocol. The assumption of eventual-faultiness yields a natural bridge between the cryptographic notion of corruption and the notion of (dis-)honestly as induced by  the trustworthiness score. Indeed, a corrupted party  who remain honest forever is no different than an honest party and has no negative effect on the security properties of the blockchain. 

We note that static reputation and party set implies that discovery of the ground truth corresponds to the so-called {\em bootstrapping} of the blockchain, in particular, creation of its parameters that can be engraved on its genesis block. Such a block is agreed upon by all parties (users and nodes) and ``engraved'' on their state at the onset of the PoR protocol. Indeed, the reputation system in~\cite{KleOstZik20} is actually part of their PoR blockchain's genesis block written on both the PoR and the (fallback) PoS  blockchain. Therefore, in the following we focus on how to perform such bootstrapping.  %for~\cite{KleOstZik20}. 

An important consideration here is that in absence of publicly agreed trustworthy reputation, agreeing on such genesis-block parameters is impossible without further assumptions. This is where the fallback PoS change becomes handy in the bootstrapping phase, as the bootstrapping happens on the PoS blockchain.

We demonstrate how to use our game(s) to perform this bootstrapping. For simplicity we focus this discussion on the perfect information game $\vgame_{\rm perfect}$. Recall that $\vgame_{\rm perfect}$ assumes that users have perfect knowledge of the state of nature. At the onset of the protocol, all users play $\vgame_{\rm perfect}$ and record their strategies on the PoS blockchain. 
The bootstrapping protocol then proceeds as follows:   we  execute  the PoR/PoS protocol, with the difference that instead of using the lottery from~\cite{KleOstZik20} for picking committees, we follow a fixed  round-robin schedule:  order all nodes lexicographically (e.g., according to their wallet address); in the first slot choose the first $c={\rm polylog}(m)$ servers,  then the next $c$ and so on. 

It is important to note that in this phase, we have no information about reputation so the safety and/or liveness of the PoR chain might  be violated in several slots. The fallback property of the PoR/PoS chain from~\cite{KleOstZik20}, guarantees that when this happens, it is detected on the secondary chain and at least some corrupted node is discovered. When this happens, we restart the bootstrapping  with this node excluded. The above process is done until some iteration completes $\lambda$ (the eventual-faultiness parameter) rounds.  

At the end of this phase, the system looks at all information on the bootstrapping phase recorded on the PoS blockchain and rewards users according to $\vgame_{\rm perfect}$, where servers detected as malicious correspond to nature sampling $0$ and the remainder as sampling $1$. Finally, the system executes the decoding function $\cal{D}_\mathrm{PR}$ and the associated reputation scores are adopted as the nodes' reputation. This completes the bootstrapping and the blockchain can then start running the protocol from~\cite{KleOstZik20} with the extracted reputation scores as its reputation system.

The following is a corollary  of Lemmas~\ref{lem:unique-nash} and~\ref{lem:eps-nash} by observing that  the assumption is that a node $s_i$ is honest with probability $\vrep_i$ and every corrupted node will be faulty (and recorded as corrupted) in the boostraping phase---this follows directly by the evantual-faultiness assumption. The above means that the rewards for parties will occur with the same probability as when nature moves in $\vgame_{\rm perfect}$.

\newtheorem{vcorollary}{Corollary}[section]

%\vassilis{something weird happens with corollary numbering}
\begin{vcorollary}
Assuming the users have perfect (resp. consistent noisy)  information on the nodes' trustworthiness, playing according to their truth-telling strategy is a Nash equilibrium (resp, $\epsilon'$-best response for $\epsilon'$ as in Lemma~\ref{lem:eps-nash}). %belief. \vassilis{do we have a name for this strategy?} \rohit{truth-telling strategy}
\end{vcorollary}

%is slit in two epochs: In the first epoch, the parties execute the Mobby with a fixed schedule of polylog(n) size committees, e.g., in a round robin fashion: order all nodes lexicographically (according to their wallet address) and in the first slot choose the first $c=polylog(n)$, then the next $c$ and so on. Since committees are chosen lexicographically, 

The decodability property of the games then ensures that the blockchain can compute an ordering of the servers according to their trustworthiness,  which has no inversions in the perfect information case and a small (negligible in $m$ when $p$ is negligible in $m$) probability of inversions in the incomplete information case.  Using this ordering, if selecting the top $\ell={\rm polylog}(m)$ nodes yields a committee where the number of corrupted parties is at most $(1/2-c'')\ell$, for some constant $c''$ (which is the assumption in~\cite{KleOstZik20}) then selecting the top $0.9\ell$ parties will also yield an honest majority committee with overwhelming probability. We defer to \cite{KleOstZik20}
 for more details. 

%We remark that in cryptography, a party that is not corrupted (also referred does not 

%of the parties are recorded on the blockchain once, at the beginning, and stay fixed for the lifetime of the system---and static 

%Intuitively, the first one rewards everyone, independent of whether or not they endorsed the chosen committee; the second one rewards only users who endorsed servers in the chosen committee, and the third one is a mix of the two. We view both mechanisms as relevant for the development of a system: In order for the rewards to be the defining factor for a user's rationality, they need to be substantial. At the early stages of a PoR system, where the focus is on adoption and most accurate discovery of the node's reputation, rewarding every users is both economically feasible (onle a few users) and more desirable (extracting reputation from more signals). However, in as the system scales, rewarding every user becomes costly, so the system would choose to resort to the second mechanism (or a mix of the two).   We note that for the purpose of this analysis, we will 

%\subsection{Rewarding Every User}

%\subsection{Rewarding Only Users who Point to Selected Committee Members}
\section{Deferred Proofs}
\label{sec:proofs}

\subsection{PageRank on Informed Cliques}
We prove that running PageRank on a perfectly informed clique yields a reputation score that is ratio-preserving with respect to the trustworthiness scores.

For simplicity, we assume $\sum_{i\in [m]}\sf{R}_i = 1$ but this can be easily generalized by normalizing with respect to L1-norm.

Let $G$ be the clique and $M$ be its adjacency matrix. Since $G$ is perfectly informed, we have the outgoing weights of each node, $\vv{w_{i}} = \inp{\sf{R}_1, \ldots, \sf{R}_m}$.
Recall that PageRank is the stationary distribution of $M$. Thus, we prove,
\begin{lemma}
\label{lem:pr-on-clique}
    The stationary distribution of $M$, $\vv{\pi}$, is ratio-preserving. 
\end{lemma}
\begin{proof}
The stationary distribution is the solution to, $$\vv{\pi} = \vv{\pi}M$$ (i.e., the eigenvector of eigenvalue 1). For $M$ in which each row $\vv{w_i}$ is identical, it is easy to verify that $\vv{\pi}$ is also in fact equal to $\inp{\sf{R_1}, \ldots, \sf{R_m}}$. Therfore, $\frac{\vv{\pi}_i}{\vv{\pi}_j} = \frac{\sf{R}_i}{\sf{R}_j}$ which concludes the proof.
\end{proof}

\subsection{Designated PageRank}

We provide a proof sketch of the fact that running Designated PageRank on the Designated Reputation Graph yields a ratio-preserving reputation score on the designated subset of reputation nodes.

\begin{lemma}
\label{lem:pr-on-moon-clique}

Let $G = (V\cup S,E)$ be a Designated Reputation graph with perfect information where $S$ is the designated reputation nodes. Then, Designated PageRank is ratio-preserving with respect to the trustworthiness of the reputation nodes.
\end{lemma}
\paragraph{Proof sketch.} Since the non-reputation nodes are perfectly informed, the outgoing weights are all equal and correspond to the true trustworthiness (L1-normalized). This implies that the incoming edges to the reputation nodes from the non-reputation nodes are all equal, and the incoming edges from other reputation nodes are 0 (by definition of designated reputation graph). Thus, similar to Lemma~\ref{lem:pr-on-clique}, the components of the stationary distribution that correspond to the reputation nodes will be in the same ratio as the incoming edges. Since the incoming edges are perfect, we recover an ratio-preserving reputation score on the reputation nodes.

% \subsection{Proofs for $\cal{G}_\rm{noisy}$}
\end{appendix}

\end{document}